\documentclass[11pt]{article}

\usepackage{amssymb}
\usepackage{amsmath}
\usepackage{epsfig}

\begin{document}
\newtheorem{theorem}{Theorem}
\newtheorem{corollary}{Corollary}
\newtheorem{conjecture}{Conjecture}
\newtheorem{definition}{Definition}
\newtheorem{lemma}{Lemma}

\newcommand{\define}{\stackrel{\triangle}{=}}

\pagestyle{empty}

\def\QED{\mbox{\rule[0pt]{1.5ex}{1.5ex}}}
\def\proof{\noindent{\it Proof: }}

\date{}
\title{On the Degrees of Freedom of Finite State Compound Wireless Networks - Settling a Conjecture by Weingarten et. al.}
\author{\normalsize  Tiangao Gou, Syed A. Jafar, Chenwei Wang\footnote{The ordering of authors is alphabetical.} \thanks{This work was supported by DARPA ITMANET under grant UTA06-793, by NSF
under grants 0546860 and CCF-0830809 and by ONR YIP under grant
N00014-08-1-0872.}\\
      {\small \it E-mail~:~\{tgou,syed,chenweiw\}@uci.edu} \\
       }
\maketitle \thispagestyle{empty}

\begin{abstract}
We explore the degrees of freedom (DoF) of three classes of finite
state compound wireless networks in this paper. First,  we study the
multiple-input single-output (MISO) finite state compound broadcast
channel (BC) with arbitrary number of users and antennas at the
transmitter. In prior work, Weingarten et. al. have found inner and
outer bounds on the DoF with 2 users. The bounds have a different
character. While the inner bound collapses to unity as the number of
states increases, the outer bound does not diminish with the
increasing number of states beyond a threshold value. It has been
conjectured that the outer bound is loose and the inner bound
represents the actual DoF. In the complex setting (all signals,
noise, and channel coefficients are complex variables) we solve a
few cases to find that the outer bound -- and not the inner bound --
of Weingarten et. al. is tight. For the real setting (all signals,
noise and channel coefficients are real variables) we completely
characterize the DoF, once again proving that the outer bound of
Weingarten et. al. is tight. We also extend the results to arbitrary
number of users. Second, we characterize the DoF of finite state
scalar (single antenna nodes) compound $X$ networks with arbitrary
number of users in the real setting. Third, we characterize the DoF
of finite state scalar compound interference networks with arbitrary
number of users in both the real and complex setting. The key
finding is that scalar interference networks and (real) $X$ networks
do not lose any DoF due to channel uncertainty at the transmitter in
the finite state compound setting. The finite state compound MISO BC
does lose DoF relative to the perfect CSIT scenario. However, what
is lost is only the DoF benefit of joint processing at transmit
antennas, without which the MISO BC reduces to an $X$ network.
\end{abstract}

\newpage
\section{Introduction}
Recent advances in network information theory -- such as the idea of
interference alignment -- have greatly widened the gap between the
theoretical capacity predictions of wireless networks and the
achievable rates with currently known practical schemes. The large
gap is reminiscent of the early days of information theory when
Shannon showed the theoretical capacity of a point to point channel
to be far beyond what was then thought achievable. The gap between
theory and practice, both then and now, can be viewed as a debate
between structured and random coding approaches. Remarkably, both
theory and practice have switched sides on the issue of random
versus structured codes. When Shannon theory advocated random
coding, practical schemes were exclusively focused on structured
codes that could be decoded with a reasonable complexity. That
debate was essentially won by theory as practical schemes like turbo
codes were found to mimic random coding schemes and thereby approach
theoretical limits with reasonable complexity. In the new debate,
the situation is reversed. Recent theoretical advances advocate
highly sophisticated structured codes while the achievable rates
considered practical draw largely on basic random coding arguments.
In the new debate it is not at all clear if theory will emerge as
the winner. The issue separating theory and practice is no longer
merely a matter of complexity at the receivers. Rather, it is not
known whether theoretical results that support structured codes
based on idealized assumptions such as perfect -- and sometimes
global -- channel knowledge at the \emph{transmitters}, will be
robust to channel uncertainty, at least to the extent that it is
fundamentally unavoidable in wireless networks. Since many of these
recent theoretical insights emerge out of the degrees of freedom
(DoF) perspective, a natural question is to explore the robustness
of the DoF results to channel uncertainty at the transmitters.

It is well known that the MIMO point to point channel and the MIMO
multiple access channel do not lose any DoF due to the lack of
channel state information at the transmitters (CSIT). Evidently,
this is because the combination of joint processing of all received
signals and perfect channel state information at the receiver (CSIR)
is able to compensate for the lack of CSIT. However, for most other
MIMO networks the DoF are not believed to be robust to channel
uncertainty at the transmiter. Consider, for example, the MIMO
broadcast channel with $M$ antennas at the transmitter and $N_1,
N_2$ antennas at the two receivers. With perfect channel knowledge
this channel has a total of $\min(M,N_1+N_2)$ DoF
\cite{Jafar_Fakhereddin}, which is the same as with perfect
cooperation at the receivers. However, in the ergodic time-varying
i.i.d. Rayleigh fading case for example, it is known that with no
CSIT, the MIMO broadcast channel loses DoF to the extent that
time-division between users is optimal
\cite{Huang_Jafar_Shamai_Vishwanath} for all points in the DoF
region. For the two user MIMO interference channel with $M_1, M_2$
antennas at the two transmitters and $N_1, N_2$ antennas at their
corresponding receivers, the DoF with perfect CSIT are characterized
in \cite{Jafar_Fakhereddin} as $\min(M_1+M_2, N_1+N_2,
\max(M_1,N_2),\max(M_2,N_1))$. With no CSIT the loss of DoF is
characterized in \cite{Huang_Jafar_Shamai_Vishwanath}.
Interestingly, the loss of DoF is shown to depend on the relative
number of antennas at the transmitters and receivers. For example,
if the transmitters have at least as many antennas as their
\emph{desired} receivers, $M_1\geq N_1, M_2\geq N_2$, then DoF are
lost to the extent that simple time-division between the two users
achieves all points in the DoF region. On the other hand, if the
receivers have at least as many antennas as their \emph{interfering}
transmitters, i.e. $N_1\geq M_2, N_2\geq M_1$ then there is no loss
of DoF due to the absence of CSIT. As in the multiple access
channel, concentration of antennas at the receivers allows the
benefits of joint signal processing under perfect channel knowledge,
which is sufficient to offset the limitations of no CSIT for the
entire DoF region. However, for most networks where the antennas are
not disproportionately located on the receivers -- such as networks
of single antenna nodes -- the DoF penalty due to the total lack of
CSIT can be quite severe. For example, under the i.i.d. fading
assumption (independent identically distributed across all
dimensions), any distributed network of single antenna nodes has
only 1 DoF in the absence of CSIT. This is because all received
signals are statistically equivalent and therefore any receiver can
decode all the messages. Since a receiver with only 1 antenna can
decode all messages, the sum DoF cannot be more than 1. The DoF loss
due to lack of CSIT is very significant for larger networks because
with full CSIT these networks have been shown to be capable of much
higher DoF. For example, an interference network with $K$
transmitter-receiver pairs is shown to have $K/2$ DoF in
\cite{Cadambe_Jafar_int}, and an $X$ network with $S$ source nodes
and $D$ destination nodes is shown to have $\frac{SD}{S+D-1}$ DoF in
\cite{Cadambe_Jafar_X}. Evidently, the transmitters' ability to
exploit the channel structure to selectively align signals -- the
key to the DoF of interference and $X$ networks -- is lost when CSIT
is entirely absent.

While perfect CSIT is an overly optimistic assumption, the complete
lack of CSIT is overly pessimistic. The collapse of DoF in the total
absence of CSIT, while sobering, is not a comprehensive argument
against the potential benefits of interference alignment in
particular or structured coding approaches in general. Hence the
need to investigate the behavior of DoF under partial channel
knowledge. Two kinds of approaches have been followed in this
regard.

The first approach investigates how the quality of CSIT should
improve as SNR increases, in order to retain the same DoF as
possible with perfect CSIT. A representative work that takes the
first approach is \cite{Caire_Jindal_Shamai} where the two user MISO
BC is investigated under the assumption that the channel vector of
one user (say user 1) is known perfectly but the channel vector of
the other user (user 2) can take one out of two values. The angular
separation $\theta$ between the two possible channel vectors of user
2 is chosen as a measure of the channel uncertainty and it is
investigated how $\theta$ should diminish as SNR approaches infinity
in order to retain the full (two) DoF possible with perfect CSIT. It
is shown that $\sin^2(\theta) = O(\mbox{SNR}^{-1})$ is required to
achieve two DoF in this setting. Other related works that follow
this approach include \cite{Yoo_Jindal_Goldsmith, Nihar_MIMOBC,
Ravindran_Jindal_MIMOBC}.

The second approach seeks the impact on DoF of a fixed amount of
channel uncertainty that is independent of SNR. References
\cite{Lapidoth_Shamai_Wigger, Weingarten_Shamai_Kramer} take this
approach for the two user MISO BC. While
\cite{Lapidoth_Shamai_Wigger} assumes channel uncertainty over a
space of non-zero probability measure under time-varying channel
conditions, \cite{Weingarten_Shamai_Kramer} investigates a finite
state compound channel setting where a specific channel state  is
drawn (unknown to the transmitter) from a finite set of allowed
states and the chosen state is held fixed throughout the duration of
communication. As large as this set may be, its finite cardinality
restricts the channel uncertainty at the transmitter to a space of
zero measure. While the two settings are quite different, the
conclusions arrived at in \cite{Lapidoth_Shamai_Wigger} and
\cite{Weingarten_Shamai_Kramer} bear striking similarities. For
example, with $M=2$ antennas at the transmitter, the best outer
bound on the DoF in both works is equal to $\frac{4}{3}$. In both
works it is conjectured that this outer bound is loose in general.
Lapidoth et. al. \cite{Lapidoth_Shamai_Wigger} conjecture that the
DoF in their setting should collapse to 1.  Remarkably, Weingarten
et. al. \cite{Weingarten_Shamai_Kramer} show the achievability of
$\frac{4}{3}$ DoF  when each user's channel can be in one of two
states\footnote{ To put this result into perspective with
\cite{Caire_Jindal_Shamai}, note that this is achieved without the
need for diminishing angular separation between the channel vectors
as SNR approaches infinity. Evidently the argument for
$\sin^2(\theta)=O(\mbox{SNR}^{-1})$, presented in
\cite{Caire_Jindal_Shamai}, is contingent on the premise that the
same DoF should be achieved as possible with perfect CSIT. }.
However, as the number of possible channel states for either user
(or both users) increases, Weingarten et. al.
\cite{Weingarten_Shamai_Kramer} also conjecture that the DoF in
their setting should collapse to 1.  Our main contribution in this
paper is to settle the latter conjecture in the negative.

The central concept involved in this work -- as well as in the
original work of Weingarten et. al. \cite{Weingarten_Shamai_Kramer}
--  is the idea of interference alignment, i.e. structuring signals
in such a way that undesired signals cast overlapping shadows where
they are not desired while they remain distinguishable where they
are desired. The idea originated out of the study of the $2$ user
$X$ channel \cite{MMK_isit,MMKreport, Jafar_Shamai}. The benefits of
overlapping interference spaces were first pointed out in the
context of the 2 user $X$ channel in \cite{MMK_isit, MMKreport} and
the concept of interference alignment was crystallized in
\cite{Jafar_Shamai}, where the first linear (based on
beamforming/zero-forcing over extended channel symbols) interference
alignment scheme was introduced. \cite{MMK_isit, MMKreport}
established the achievability of $\lfloor\frac{4}{3}M\rfloor$ DoF
for the $2$ user $X$ channel where all nodes are equipped with $M$
antennas each. For the same channel model, \cite{Jafar_Shamai}
derived the outer bound, DoF $\leq \frac{4}{3}M$, and also proved
its achievability  for $M>1$. For $M=1$ (single antenna at all
nodes), \cite{Jafar_Shamai} showed the achievability of
$\frac{4}{3}$ DoF  only for a time-varying/frequency-selective
channel model. Achievability of $\frac{4}{3}$ DoF for the constant
(not time-varying or frequency-selective) channel case, with complex
channel coefficients, was established in \cite{Cadambe_Jafar_Wang}
where the idea of asymmetric complex signaling was introduced to
achieve interference alignment. The constant channel case with real
channel coefficients was studied in \cite{Motahari_Gharan_Khandani}.
Building on the idea of  alignment of lattices scaled by
rational/irrational factors, originally introduced in
\cite{Etkin_Ordentlich}, it was shown by
\cite{Motahari_Gharan_Khandani} that the outer bound of
$\frac{4}{3}$ DoF is also achievable when the channel coefficients
are real. All these ideas -- zero-forcing, beamforming, channel
extensions, asymmetric complex signaling, rational/irrational scaled
lattice alignment -- are used for the achievable schemes in this
paper. In addition, we make use of the interference alignment
schemes used for the SIMO interference channel in
\cite{Gou_Jafar_SIMO}, which turns out to be the dual/reciprocal
network for the compound MISO BC. An interesting outcome of this
duality perspective is to clarify the role of alignment of vector
spaces at the \emph{transmitter} instead of the receivers.

A striking observation from the results summarized above, is the
recurrence of the fraction $\frac{4}{3}$  in the DoF
characterizations of both the 2 user $X$ channel, as well as the
compound MISO BC. As we find in this work,  this is not merely a
coincidence. With enough channel uncertainty, the finite state
compound BC -- regardless of the number of users or transmit
antennas -- loses the DoF benefits of joint signal processing at the
transmitter. From the DoF perspective, this makes the finite state
compound MISO BC equivalent to a finite state compound $X$ channel.
Moreover, the finite state compound $X$ channel  does not lose any
DoF compared to the perfect CSIT scenario (non-compound setting).
Thus, the DoF of the compound MISO BC end up being equal to the DoF
of the $X$ channel obtained by separating the transmit antennas.
Similar to $X$ networks, we find that $K$ user interference networks
also do not lose DoF in the finite state compound channel setting.
It should be noted that some of these results are found in the real
setting, i.e. all channel coefficients and signals and noise are
restricted to take only real values.

We present the system model in the next section. The main results
are presented as theorems in Section \ref{sec:complexmiso} and
\ref{sec:4} along with the main ideas needed for the proofs. The
detailed proofs are presented in the Appendix and the conclusions
are summarized in Section \ref{sec:5}.

\section{Compound MISO Broadcast Channel - Complex Setting}\label{sec:complexmiso}
A compound MISO  broadcast channel consists of a transmitter with
$M>1$ antennas and $K$ single antenna receivers. The channel vector
 $\mathbf{h}^{[k]}$ associated with user $k$  is drawn from a set
$\mathcal{J}_k$ with finite cardinality $J_k$. To avoid degenerate
cases, we assume the channel states  are drawn from a continuous
distribution. Thus, almost surely the channel states are generic,
e.g., the coefficients are algebraically independent. Once the
channel is drawn, it remains unchanged during the entire
transmission. While the transmitter is unaware of the specific
channel state realization, the receivers are assumed to have perfect
channel knowledge. The transmitter sends independent messages
$W^{[k]}$ with rates $R^{[k]}$ to receiver $k=1,2,\ldots, K$,
respectively. A rate tuple $(R^{[1]},R^{[2]},\ldots, R^{[K]})$ is
achievable if each receiver is able to decode its message with
arbitrary small error probability regardless of state (realization)
of the channel. The received signal of user $k$ corresponding to
channel state index $j_k$ is given by
\begin{eqnarray}
y^{[k]}_{j_k}(n)= \mathbf{h}^{[k]}_{j_k}{\bf
x}(n)+z^{[k]}_{j_k}(n)~~~~k=1,\ldots,K~~ j_k=1,\ldots, J_k
\end{eqnarray}
$\mathbf{h}^{[k]}_{j_k}=\left[h_{j_k1}^{[k]},\ldots,h_{j_kM}^{[k]}\right]$
is a $1\times M$ channel vector between the transmitter and receiver
$k$ under state $j_k$ where $j_k\in\{1,\ldots,J_k \}$.
$\mathbf{x}=\left[x_1(n),\ldots,x_M(n)\right]^T$ is an $M \times 1$
transmitted complex vector at time $n$ and satisfies the average
power constraint $E(\|\mathbf{x}\|^2)\leq P$. $z^{[k]}_{j_k}$
represents independent identically distributed (i.i.d.) zero mean
unit variance circularly symmetric complex Gaussian noise. The total
number of degrees of freedom $d$ is defined as
\begin{eqnarray}
d=\lim_{P \rightarrow \infty}\frac{R^{[1]}+\cdots+R^{[K]}}{\log P}
\end{eqnarray}
A two user compound broadcast channel with $M=2,J_1=J_2=J=2$ is
shown in Figure \ref{fig:2usercompoundbc}.

\begin{figure}[!t]
\centering
\includegraphics[width=3in]{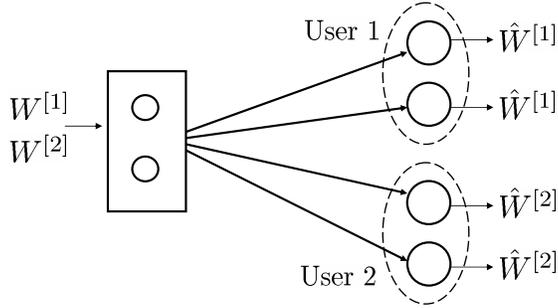}
\caption{$2$ User Compound Broadcast Channel with $J_1=J_2=2$}
\label{fig:2usercompoundbc}
\end{figure}

{\it Remark: } The compound broadcast channel is equivalent to a broadcast channel with
common messages. This can be seen by considering different states as
different users. Now instead of a  $K$ user compound broadcast
channel, we have a  $J_1+\cdots+J_K$ user broadcast channel with $K$
common messages, one for each group $k, \forall k=1,\ldots, K$, with
$J_k$ users.

\subsection{Degrees of Freedom of the Complex Compound MISO BC}
The degrees of freedom of the complex compound MISO BC are studied
by Weingarten, Shamai and Kramer in \cite{Weingarten_Shamai_Kramer}.
The exact DoF are found for some cases and conjectures are made for
more general scenarios. The achievability of the conjectured DoF is
established in \cite{Weingarten_Shamai_Kramer}. We start with the
first conjecture, re-stated here in the terminology of our system
model.

\subsection*{Case 1: $J_1=1,J_2=J\geq M$}
\begin{conjecture} \label{conjecture1}
(Weingarten et. al. \cite{Weingarten_Shamai_Kramer}) Consider a
complex compound BC with $K=2$ users, $M$ antennas at the
transmitter, and $J_1=1, J_2=J\geq M$ possible generic states for
users 1,2 respectively. Then the total number of DoF is
$1+\frac{M-1}{J}$, almost surely.
\end{conjecture}

Consider the MISO BC with 2 antennas at the transmitter. The case of
perfect CSIT corresponds to $J_1=J_2=1$. In this case it is easy to
see that 2 DoF can be achieved by zero forcing at the transmitter.
Specifically, the transmitter beamforms to user 1 in a direction
orthogonal to the channel vector of user 2, and beamforms to user 2
in a direction orthogonal to the channel vector of user 1. Since
neither user sees interference, they are able to achieve 1 DoF each.

Now, let us introduce some channel uncertainty with $J_1=1,
J_2=J=2$, i.e. user 1's channel is perfectly known to the
transmitter but user 2's channel can take one out of two known
values. In this case it is clear that the transmitter can still
choose a beamforming direction to user 2 that is orthogonal to the
known channel vector of user 1. However, it is not possible to pick
a beamforming vector for user 1 that is orthogonal to both the
possible channel vectors of user 2. This is because the transmitter
has only 2 antennas and the two possible channel vectors of user 2
span the entire two dimensional transmit space available to the
transmitter. It is shown by Weingarten et. al.
\cite{Weingarten_Shamai_Kramer} that the best thing to do in this
setting, from a DoF perspective, is to choose the beamforming vector
for user 1 to be orthogonal to the first possible channel vector of
user 2 for half the time and then choose it to be orthogonal to the
second possible channel vector of user 2 for the remaining half of
the time. In this manner, regardless of his state, user 2 is able to
see an interference free channel for half the time, thus achieving
0.5 DoF. At the same time, user 1 sees no interference from user 2's
signal and is able to achieve 1 DoF. Thus a total of $\frac{3}{2}$
DoF is achieved. Following the same idea, one can achieve
$1+\frac{M-1}{J}$ DoF for general values of $M, J$. Interestingly,
when $J$ is equal to $M$, \cite{Weingarten_Shamai_Kramer} shows that
this is optimal. Thus the compound BC loses DoF relative to the
perfect CSIT setting. For $J>M$ it is conjectured that
$1+\frac{M-1}{J}$ is still optimal. Note that if this conjecture
were to be true, this would mean that the DoF of the MISO BC
collapse to $1$ as the number of channel states of either user
increases.

To disprove this conjecture, we provide a specific counter example
in the following theorem.

\begin{theorem}\label{theorem:firstconj}
For the complex compound MISO BC with $K=2$ users, $M=2$ antennas at
the transmitter and $J_1=1$, $J_2=J=3$ generic channel states for
users $1,2$ respectively, the exact number of total DoF =
$\frac{3}{2}$, almost surely.
\end{theorem}

Since $1+\frac{M-1}{J}=\frac{4}{3}<\frac{3}{2}$ in this case,
Conjecture \ref{conjecture1} is disproved by Theorem
\ref{theorem:firstconj}. Interestingly, Theorem
\ref{theorem:firstconj} indicates that the total number of DoF does
\emph{not} decrease even as the number of possible channel states
for user 2 increases from 2 to 3. The reason we are able to achieve
more than the conjectured DoF in this case, can be attributed to
interference alignment schemes inspired by recently developed
insights on asymmetric complex signaling \cite{Cadambe_Jafar_Wang},
and a reciprocity/duality relationship with the interference
alignment problem for SIMO interference channels
\cite{Gou_Jafar_SIMO} that offers a clear perspective of
interference alignment at the \emph{transmitter}.

Asymmetric complex signaling involves reducing the complex setting
to a real setting by treating a complex number as a two dimensional
vector with real elements. In this setting, the scenario of Theorem
\ref{theorem:firstconj} can be seen as a transmitter with $4$
antennas, while each receiver has two antennas. Very importantly,
this mapping to the real setting introduces some structure in the
channel as complex channel coefficients are translated to
quaternionic matrix forms. Because of this structure the proof of
Theorem \ref{theorem:firstconj} involves some subtleties that may
distract from the main concept behind the interference alignment.
For simplicity of exposition, we defer the fine details of the proof
to Appendix \ref{app:firstconj} and highlight the main concepts in
this section. Specifically, in the following explanation we ignore
the special structure of the channel matrices and treat them as
generic MIMO channels between a $4$-antenna transmitter and
$2$-antenna receivers.

\begin{figure}[!t]
\centering
\includegraphics[width=5.5in]{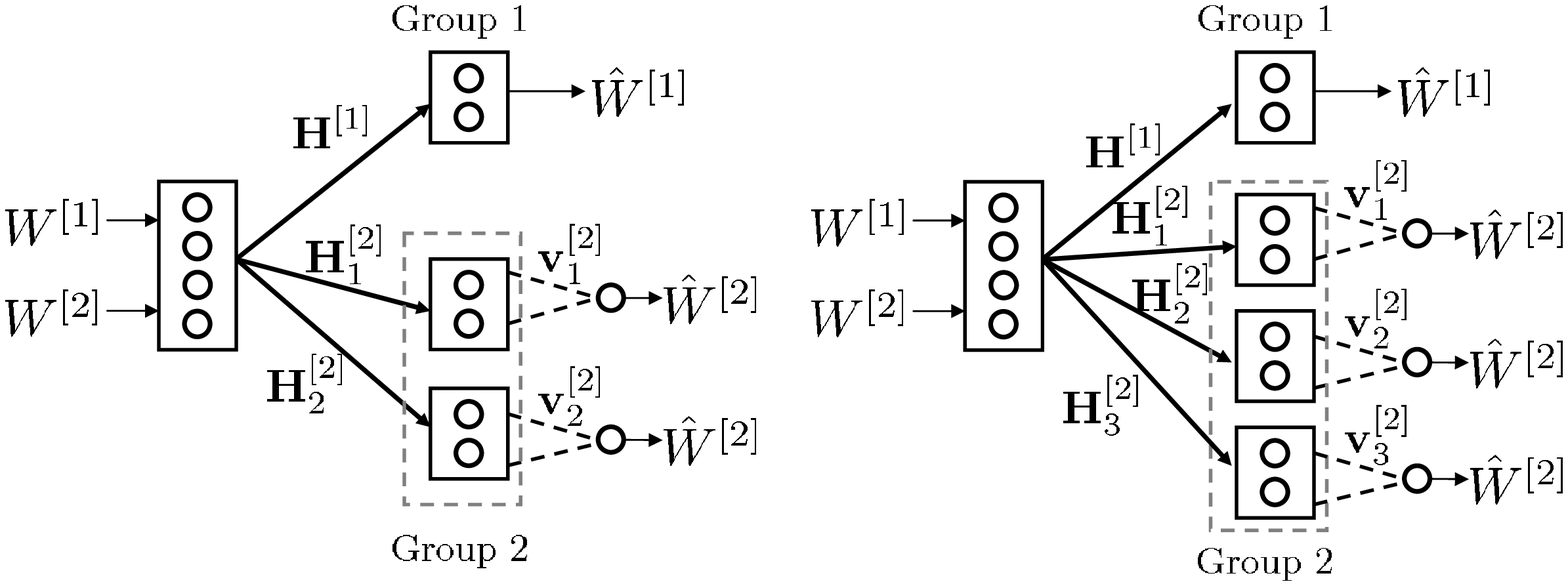}
\caption{$2$ user compound MIMO BC with $J_1=1,J_2=2$ and
$J_1=1,J_2=3$}\label{fig:compoundmimo}
\end{figure}

As mentioned before, the 2 user compound MIMO BC is equivalent to a
MIMO BC with two common messages, one for each of the two groups.
Group $k=1,2,$ consists of $J_k$ users corresponding to $J_k$ states
of user $k$ in the compound BC, and the users in the same group need
to decode the same message (see Figure {\ref{fig:compoundmimo}}).
Since there is only one receiver in group 1, we omit the index $j_1$
and replace $j_2$ with $j$ for simplicity. In addition, we use ${\bf
H}^{[1]}, {\bf H}_j^{[2]}$ to denote the channel matrix from the
transmitter to user 1 (group 1) and receiver $j$ in group 2,
respectively.

Consider first, an alternate achievable scheme for the case of
$J_1=1,J_2=2$. We let the two receivers in group 2 use arbitrarily
picked combining column vectors ${\bf v}_1^{[2]}$ and ${\bf
v}_2^{[2]}$, respectively, along which they require interference
free reception in order to achieve 1 DoF for their desired message.
In order to protect these group 2 receivers, the transmitter sends
user 1's message along the directions orthogonal to ${\bf
v}_1^{[2]T}{\bf H}_1^{[2]}$ and ${\bf v}_2^{[2]T}{\bf H}_2^{[2]}$.
Since the transmitter has $4$ antennas -- i.e., a $4$ dimensional
transmit signal space -- it is able to find a two dimensional
subspace that is orthogonal to the protected dimensions of user $2$.
This allows $2$ real streams to be sent to receiver 1 that do not
interfere with the chosen directions of any of the receivers in
group 2. On the other hand user 2's message will be sent along the
null space of ${\bf H}^{[1]}$. Since this is a rank $2$ matrix, it
also has a two dimensional null space along which user 2's message
can be transmitted. However, because each of group 2 receivers have
chosen only one receive dimension, only one interference free stream
is sent. Thus, a total of $3$ interference free streams are
delivered (2 streams to user 1 and 1 stream to every receiver of
group 2). Since these are real signals, the total DoF achieved is
$\frac{3}{2}$.

Now consider the case $J_2=3$. Suppose the third user in group 2
chooses a combining vector ${\bf v}_3^{[2]}$ for its
interference-free desired signal dimension. Now in order to protect
the chosen dimensions of group 2 receivers, the signals for user 1
should be transmitted orthogonal to the three $1\times 4$  vectors
${\bf v}_1^{[2]T}{\bf H}_1^{[2]}$, ${\bf v}_2^{[2]T}{\bf H}_2^{[2]}$
and ${\bf v}_3^{[2]T}{\bf H}_3^{[2]}$. Without alignment these three
vectors will span a three dimensional space that needs to be
protected from user 1's signal, leaving only 1 dimension at the
transmitter to send user 1's message. However, we wish to allow user
1 to still access 2 interference-free dimensions. To achieve this
goal, we need to make these three vectors span only a two
dimensional vector subspace as seen by the transmitter, i.e. the
three vectors should be linearly dependent. Equivalently, three
column vectors ${\bf H}_1^{[2]T}{\bf v}_1^{[2]},{\bf H}_2^{[2]T}{\bf
v}_2^{[2]},{\bf H}_3^{[2]T}{\bf v}_3^{[2]}$ should be linearly
dependent. Since ${\bf H}_1^{[2]T}$, ${\bf H}_2^{[2]T}$ and ${\bf
H}_3^{[2]T}$ are three $4\times 2$ generic matrices, the column
spaces of any two of the three matrices only have null intersection,
almost surely. In other words, any two of ${\bf H}_1^{[2]T}{\bf
v}_1^{[2]},{\bf H}_2^{[2]T}{\bf v}_2^{[2]}$ and ${\bf
H}_3^{[2]T}{\bf v}_3^{[2]}$ cannot be aligned along the same
direction. Therefore, we align the vector ${\bf H}_3^{[2]T}{\bf
v}_3^{[2]}$ in the space spanned by ${\bf H}_1^{[2]T}{\bf
v}_1^{[2]}$ and ${\bf H}_2^{[2]T}{\bf v}_2^{[2]}$. Mathematically,
we have
\begin{eqnarray}
{\bf H}_3^{[2]T}{\bf v}_3^{[2]}\in \textrm{span}([{\bf
H}_1^{[2]T}{\bf v}_1^{[2]}~{\bf H}_2^{[2]T}{\bf v}_2^{[2]}])\label{eq:simoinsight}
\end{eqnarray}

The above alignment is the key to achieving more than the
conjectured DoF in this case. Whereas most interference-alignment
schemes \cite{Cadambe_Jafar_int, Cadambe_Jafar_X, Jafar_Shamai,
Cadambe_Jafar_Wang} align interfering dimensions in a one-to-one
fashion, this is impossible in this case as pointed above. What is
needed instead, is an alignment of an interfering dimension within
the space spanned by several others. As it turns out, this problem
bears a striking resemblance to the interference alignment scheme
used in \cite{Gou_Jafar_SIMO} for the MIMO interference channel
where the interference vectors from one interferer are aligned in
the space spanned by the interference vectors of other interferers.

To illustrate the concept with an example, let us consider a $4$
user many-to-one interference network with $2$ antennas at each
transmitter and $4$ antennas at each receiver (see Figure
{\ref{4userifc}}). We claim user 1 can achieve $2$ DoF and other
users achieve 1 DoF simultaneously. Let ${\bf H}^{[ji]}$ denote the
$4\times 2$ channel from transmitter $i$ to receiver $j$,${\bf
V}^{[1]}$ denote the $2\times 2$ beamforming matrix for transmitter
1 and ${\bf v}^{[i]}$ denote the $2\times 1$ beamforming vectors of
transmitter $i=2,3,4$. Choose ${\bf v}^{[4]}$ randomly and let
\begin{eqnarray}
\left[\begin{array}{c}{\bf v}^{[2]}\\{\bf
v}^{[3]}\end{array}\right]={\left[{\bf H}^{[12]}~{\bf
H}^{[13]}\right]}^{-1}{\bf H}^{[14]}{\bf v}^{[4]}\Rightarrow{\bf
H}^{[14]}{\bf v}^{[4]}\in span(\left[{\bf H}^{[12]}{\bf
v}^{[2]}~{\bf H}^{[13]}{\bf v}^{[3]}\right]) \label{eqn8}
\end{eqnarray}

\begin{figure}[!b]
\centering
\includegraphics[width=4.5in, trim= 0 0 0 -40]{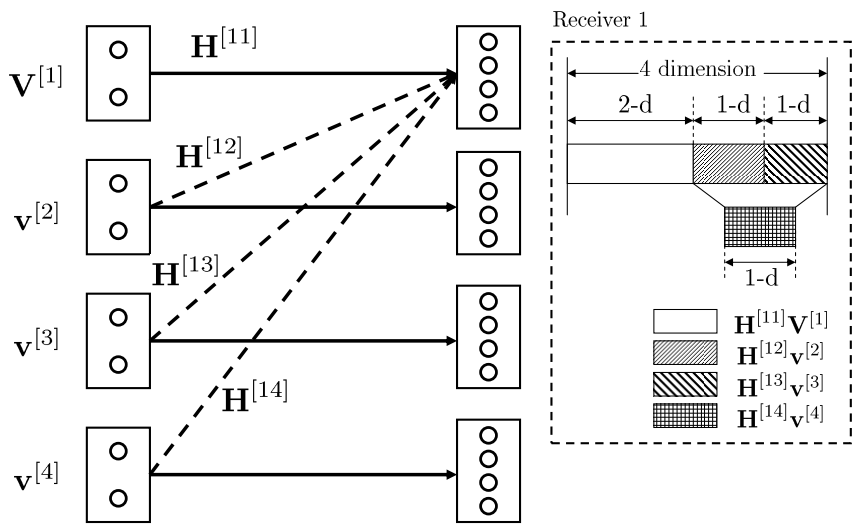}
\caption{$4$ user many-to-one interference network with 2 antennas
at transmitter, 4 antennas at receiver} \label{4userifc}
\end{figure}

\begin{figure}[!t]
\centering
\includegraphics[width=4in]{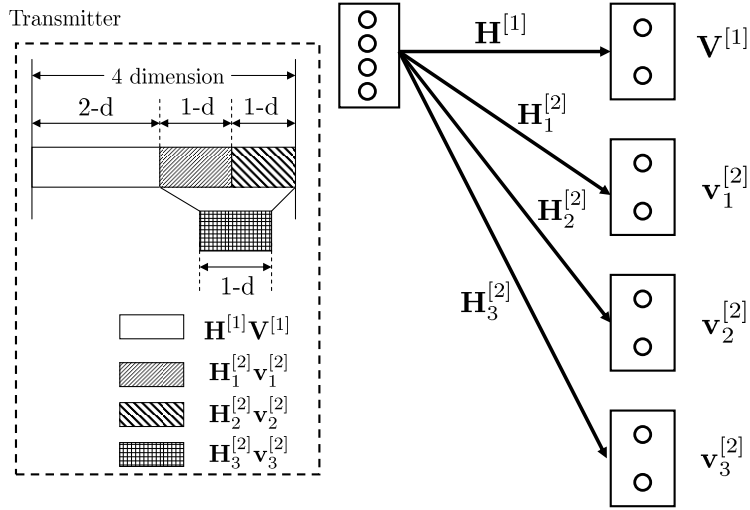}
\caption{$2$ compound broadcast channel with 4 antennas at
transmitter, 2 antennas at receiver, and $J_1=1, J_2=J=3$}
\label{compoundbcmimo}
\end{figure}

Since receivers $2$ to $4$ are interference free they can decode
their own message successfully. Now consider receiver 1. The spaces
spanned by the column vectors of ${\bf H}^{[12]}$ and ${\bf
H}^{[13]}$ only have null intersection. Thus the interference from
user $2$ and $3$ together occupies $2$ dimensions -- i.e., it does
not align. For the same reason, the interference from user 4 can
also not be aligned within the one dimensional interference from
user 2 or from user 3, individually. However,  the interference from
user $4$ is aligned in the 2 dimensional subspace space spanned by
interference from user $2$ and $3$ together. Thus, receiver $1$ sees
$4-2=2$ interference free dimensions for its intended signal
vectors. Comparing Figure {\ref{compoundbcmimo}} with Figure
{\ref{4userifc}}, the only difference is that the signal spaces are
aligned at the {\em transmitter} for compound broadcast channel
while they are aligned at the receiver for many-to-one interference
channel. In other words, from the viewpoint of signal vector
alignment, this two user compound broadcast channel is a reciprocal
version of the many-to-one interference network.

Using the insights from the above reciprocity, the solution of the
alignment problem in (\ref{eq:simoinsight}) is immediately obvious.
It is accomplished by setting
\begin{eqnarray}
\left[\begin{array}{c}{\bf v}_1^{[2]}\\{\bf
v}_2^{[2]}\end{array}\right]={\left[{\bf H}_1^{[2]T}~{\bf
H}_2^{[2]T}\right]}^{-1}{\bf H}_3^{[2]T}{\bf v}_3^{[2]}
\end{eqnarray}

The remaining details of the proof -- including the impact of
channel structure in this case -- are presented in Appendix
\ref{app:firstconj}. Finally, note that the converse is trivial here
because $\frac{3}{2}$ DoF are already shown to be optimal for
$J_2=2$ and DoF cannot increase with increasing channel uncertainty.

\subsection*{Case 2: $J_1=J_2=J\geq M$}
The second case captures the setting where the channel states of
both users are unknown to the transmitter.
\begin{conjecture} \label{conjecture2}
(Weingarten et. al. \cite{Weingarten_Shamai_Kramer}) Consider a
complex compound BC with $K=2$ users, $M$ antennas at the
transmitter, and $J_1=J_2=J\geq M$ possible generic states for users
1,2 respectively. Then the total number of DoF is
$\frac{2J}{2J-M+1}$, almost surely.
\end{conjecture}

For the case that $J=M$, this conjecture is shown to be tight in
\cite{Weingarten_Shamai_Kramer}. Note that the collapse of DoF to
unity as $J$ increases, also evident here, is already implied by
Conjecture \ref{conjecture1}.

An important observation here is that the achievability of the
$\frac{2J}{2J-M+1}$ DoF for $J\geq M$ already requires interference
alignment and is quite non-trivial. In fact this is the first
application of the concept of interference alignment outside the 2
user $X$ channel. We explain the need for interference alignment as
follows.

\begin{figure}[!t]
\centering
\includegraphics[width=4in]{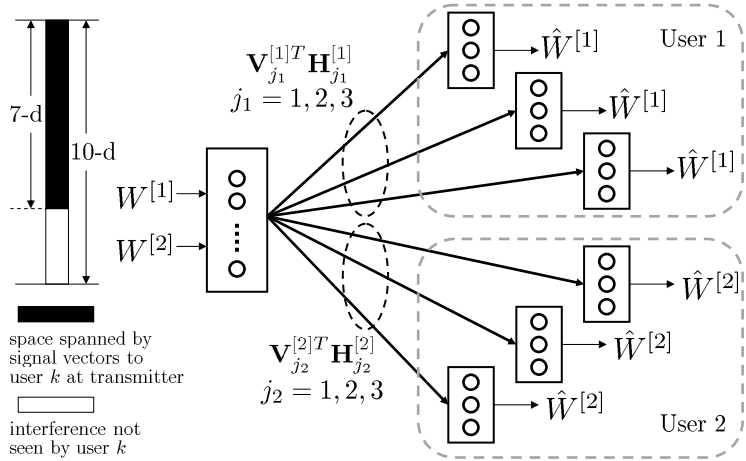}
\caption{$2$ user MISO compound broadcast channel, $M=2,J_1=J_2=3$,
5 channel extensions} \label{compoundj3}
\end{figure}

Consider a complex compound MISO broadcast channel with $K=2$ users,
2 antennas at transmitter, 1 antenna at each receiver, and
$J_1=J_2=J=3$ fading states. It is shown in Theorem 7 of
\cite{Weingarten_Shamai_Kramer} that a total of
$\frac{2J}{2J-M+1}=\frac{6}{5}$ DoF can be achieved in this case.
This is done by coding over 5 consecutive time slots to achieve 3
DoF for each user. With 5 time slots, the original $1\times 2$ MISO
channel ${\bf h}_{j_k}^{[k]}$ is converted to a $5\times 10$ MIMO
channel ${\bf H}_{j_k}^{[k]}$ but the channel has a block diagonal
structure, i.e. ${\bf H}_{j_k}^{[k]}={\bf h}_{j_k}^{[k]}\otimes {\bf
I}_{5\times 5}$, where $\otimes$ indicates the Kronecker product
operation. For fading state $j_k=1,2,3$ of user $k=1,2$, a $5\times
3$ linear combining matrix ${\bf V}_{j_k}^{[k]}$ is used whose
column vectors are respectively chosen from 3 columns of a $5\times
5$ identity matrix for user 1 and DFT matrix for user 2 such that as
seen by the transmitter the space spanned by column vectors of
$\left[{\bf H}_1^{[k]T}{\bf V}_1^{[k]}~{\bf H}_2^{[k]T}{\bf
V}_2^{[k]}~{\bf H}_3^{[k]T}{\bf V}_3^{[k]}\right]$ has 7 dimensions
(see Figure {\ref{compoundj3}}). In other words, from the
transmitter's point of view, the total number of dimensions to be
protected for user $k$ is equal to 7. Since the transmitter has
access to 10 dimensions, it can send 3 data streams to each user
along directions orthogonal to the dimensions occupied by the other
user. Thus, each user can get 3 interference free data streams and a
total of $\frac{6}{5}$ DoF is achieved per channel use. Note that if
${\bf V}_{j_k}^{[k]},j_k=1,2,3~k=1,2$ are generated randomly, the
column space spanned by $\left[{\bf H}_1^{[k]T}{\bf V}_1^{[k]}~{\bf
H}_2^{[k]T}{\bf V}_2^{[k]}~{\bf H}_3^{[k]T}{\bf V}_3^{[k]}\right]$
would have 9 dimensions. Therefore, interference alignment is the
key to the achievable scheme of \cite{Weingarten_Shamai_Kramer}.

It turns out, this is not the most efficient interference alignment
scheme. The following theorem disproves Conjecture \ref{conjecture2}
through another counter example.

\begin{theorem}\label{theorem:secondconj}
For the complex compound MISO BC with $K=2$ users, $M=2$ antennas at
the transmitter and $J_1=J_2=J=3$ generic channel states for each
user, the exact number of total DoF = $\frac{4}{3}$, almost surely.
\end{theorem}

Since $\frac{2J}{2J-M+1}=\frac{6}{5}<\frac{4}{3}$, Conjecture
\ref{conjecture2} is disproved by Theorem \ref{theorem:secondconj}.
Once again, Theorem \ref{theorem:secondconj} indicates that the
total number of DoF does \emph{not} decrease as the number of
possible channel states $J$ for each user increases from 2 to 3.

The proof of Theorem \ref{theorem:secondconj} is based on asymmetric
complex signaling over multiple channel uses and is deferred to
Appendix \ref{app:secondconj}. Except for the detailed nuances
required to deal with the channel structure imposed by channel
extensions, the essence of the proof follows from the same
interference alignment ideas outlined in the previous section for
Theorem \ref{theorem:firstconj}.

Theorem \ref{theorem:firstconj} and \ref{theorem:secondconj} present
only specific counter examples to disprove Conjectures
\ref{conjecture1} and \ref{conjecture2}. In both cases the DoF are
shown to remain unchanged as the number of possible states for one
or both users is increased by one. From these results it is still
not clear what happens as the number of states continues to
increase. The problem with extending the results above lies with the
limitations of the linear interference alignment approach when
channel values are held constant. The difficulty is similar to the
problem of characterization of the DoF of the $K$ user interference
channel. Linear alignment solutions were shown in that case to
achieve the $\frac{K}{2}$ outer bound under
time-varying/frequency-selective channel conditions
\cite{Cadambe_Jafar_int}. For the constant channel case, even though
asymmetric complex signaling was able to achieve $\frac{6}{5}$ (i.e.
greater than one) DoF, it was still away from the $\frac{K}{2}$
outer bound. Interestingly, lattice alignment schemes were needed to
show the achievability of $\frac{K}{2}$ DoF almost surely
\cite{Etkin_Ordentlich, Motahari_Gharan_Khandani}. For the compound
MISO BC as well, it turns out that complete DoF characterizations
can be found using lattice alignment schemes in the real setting,
i.e. where all channel coefficients, signals and noise terms are
real variables.

\section{Compound MISO Broadcast Channel - Real Setting}{\label{sec:4}}
We consider the real compound broadcast channel for which the
channel input-output relationship is similar to the complex case but
with the channel, input signal and noise terms restricted to real
values. In the real setting, the total number of degrees of freedom,
$d$, is defined as
\begin{eqnarray}
d=\lim_{P \rightarrow
\infty}\frac{R^{[1]}+\cdots+R^{[K]}}{\frac{1}{2}{\log P}}
\end{eqnarray}
Note that in the previous section, we solved the interference
alignment problem for the complex compound BC only by viewing the
complex variables as two dimensional real vectors. Therefore it may
not be clear why the real setting should be considered separately.
The answer lies in the structure of the channel. Translating the
complex setting into the real setting, as mentioned before, imposes
a special structure on the channel because complex scalar
coefficients get replaced with quaternionic matrices. However, in
this section we will assume \emph{generic} real channel
coefficients, i.e. the channel coefficients will, almost surely, be
algebraically independent over rationals. The connections to the
complex setting will be discussed in Section \ref{sec:discussion}.

\subsection{Interference Alignment in rational dimensions}

It is well known that a multi-dimensional signal space provides
multiple independent signalling dimensions. By communicating along
linearly independent (beamforming) vectors, different data streams
can be separated. Moreover, in a multiuser communication network
where interference exists, linear independence between desired
signal and interference can be used to separate them as well. The
number of DoF is essentially equal to the number of interference
free dimensions. Thus, to maximize the achievable DoF, we should
minimize the dimension occupied by interference. This is the idea of
linear interference alignment, which is exploited to align
interference in signal space provided by spatial/time/frequency
dimensions \cite{Cadambe_Jafar_int, Cadambe_Jafar_X}.

For a network with real constant channel coefficients and single
antenna nodes, the notion of signal level as a dimension is very
useful. Alignment in this dimension is achieved through multi-level
lattice codes, e.g. \cite{Bresler_Parekh_Tse,
Cadambe_Jafar_Shamai,Etkin_Ordentlich,Motahari_Gharan_Maddahali_Khandani}.
Recent work by Etkin and Ordentlich in \cite{Etkin_Ordentlich} and
by Motahari et. al. \cite{Motahari_Gharan_Maddahali_Khandani,
Motahari_Gharan_Khandani} shows that interference alignment can be
exploited in signal scale dimension based on the notion of rational
independence. In this case, different data streams are multiplexed
using {\em rationally independent} coefficients. In fact, rationally
independent coefficients in  scalar channels play the same role as
linearly independent vectors in vector channels. They serve as
distinct directions along which several data streams can be carried
simultaneously and can be exploited to separate interference and
desired signals as well. In addition, similar to the case in signal
space, we can determine the number of DoF by simply counting the
number of interference free rational dimensions. Instead of
providing 1 DoF per dimension in the signal space, in an $m$
dimensional rational space each dimension can carry $\frac{1}{m}$
degrees of freedom if certain conditions are satisfied. Intuitively,
this is because for a 1 dimensional signal space, only 1 DoF is
available, and hence each rational dimension can carry $\frac{1}{m}$
DoF.

Next, we summarize the conditions in
\cite{Motahari_Gharan_Maddahali_Khandani} under which each data
stream can achieve $\frac{1}{m}$ DoF in interference networks to
multiuser wireless networks where $m$ denotes the maximum number of
rational dimensions received among all receivers. As in
\cite{Motahari_Gharan_Maddahali_Khandani}, we denote a set of
monomials with variables from a set of algebraically independent
numbers over rational numbers as $\mathcal{G}(\mathbf{h})$. In other
words, a member of $\mathcal{G}(\mathbf{h})$ is of the form
$h_1^{\alpha_1}h_2^{\alpha_2} \cdots h_n^{\alpha_n}$ where
$\alpha_1,\ldots, \alpha_n$ are nonnegative integers and
$h_1,\ldots, h_n$ are algebraically independent over the rational
numbers. Note that all members in $\mathcal{G}(\mathbf{h})$ are
rationally independent.

Consider a multiuser wireless network with real channel coefficients
where there are $S$ transmitters and $D$ receivers. Each transmitter
may have a message for each receiver. For any $\epsilon>0$,
transmitter $i, \forall i=1,2,\ldots, S,$ generates $D_i$
independent data streams by uniformly picking up integers from
interval
$(-P^{\frac{1-\epsilon}{2(m+\epsilon)}},P^{\frac{1-\epsilon}{2(m+\epsilon)}})$.
Essentially, each data stream carries $\frac{1}{m}$ DoF. Then these
data streams are multiplexed by rationally independent numbers
$V_{i0}, V_{i1},\ldots, V_{i(D_i-1)}$ which serve as distinct
directions. In order to satisfy the power constraint, the signal is
transmitted with a scaling factor $A=\lambda
P^{\frac{m-1+2\epsilon}{2(m+\epsilon)}}$ where $\lambda$ is a
constant. Now, suppose at receiver $j$, there are $L_j$ desired data
streams received along directions $V'_{j0}, V'_{j1},\ldots,
V'_{j(L_j-1)}$ and interference data streams are received in a
$L'_{j}$ dimensional space over rational numbers with a basis
$U_{j0}, U_{j1}, \ldots, U_{j(L'_j-1)}$. In other words, there are
$L'_j$ effective interference data streams along directions $U_{j0},
U_{j1}, \ldots, U_{j(L'_j-1)}$. Each data stream can almost surely
achieve a rate $\frac{1}{2m}\log P +o(\log P)$ and hence
$\frac{1}{m}$ degrees of freedom where $m$ is the maximum number of
rational dimensions received among all receivers, i.e., $m=\max_j
L_j+L'_j$, if following conditions are satisfied:
\begin{enumerate}
\item $V_{i0}, V_{i1},\ldots, V_{i(D_i-1)}$ are distinct members of
$\mathcal{G}(\mathbf{h})$.
\item $V'_{j0}, V'_{j1},\ldots,
V'_{j(L_j-1)}$, $U_{j0}, U_{j1}, \ldots, U_{j(L'_j-1)}$ are all
distinct.
\item One of $V'_{j0}, V'_{j1},\ldots,
V'_{j(L_j-1)}$, $U_{j0}, U_{j1}, \ldots, U_{j(L'_j-1)}$  is 1.
\end{enumerate}
Note that the first condition ensures that $V_{i0}, V_{i1},\ldots,
V_{i(D_i-1)}$ are rationally independent and the second condition
ensures that the desired signals and interference are rationally
independent so that they can be separated. Along with the first and
second condition, the third condition can be used to show that the
distance between any two points in the receive constellation grows
with $P$ \cite{Motahari_Gharan_Maddahali_Khandani}. Thus, at high
SNR, the message can be decoded with arbitrary small probability. In
addition, as in \cite{Motahari_Gharan_Maddahali_Khandani}, if none
of $V'_{j0}, V'_{j1},\ldots, V'_{j(L_j-1)}$, $U_{j0}, U_{j1},
\ldots, U_{j(L'_j-1)}$  is 1, then $\frac{1}{m+1}$ degrees of
freedom can be achieved for each data stream.

We can see that to maximize the achievable DoF, the key is to
minimize the dimensions of the space spanned by interference. Note
that here the space denotes the set of all linear combinations of
rationally independent numbers with {\em rational coefficients}.
Ideally, we wish to align interference from different users
perfectly with each other. For example, if the interference received
at one receiver from the $i$th user is along members of
$h_i\mathbf{V}_i, \forall i=1,\ldots, N$, we wish to align them as
$\textrm{span}(h_1\mathbf{V}_1)=\textrm{span}(h_2\mathbf{V}_2)=\cdots=\textrm{span}(h_N\mathbf{V}_N)$
where $\textrm{span}(A)$ denotes the space spanned by columns of
$A$. However, it turns out that such alignment is infeasible in
general. In fact, similar problem appears in vector space alignment
for interference networks and wireless $X$ networks as well
\cite{Cadambe_Jafar_int, Cadambe_Jafar_X}. Fortunately, as shown in
\cite{Cadambe_Jafar_int, Cadambe_Jafar_X} alignment is feasible if
we allow a negligible fraction of interference terms not aligned
perfectly. Due to the similarity of spatial dimensions and rational
dimensions, the vector space alignment schemes can be directly
translated into alignment schemes in rational dimensions. We present
the idea in the following lemma.

\begin{lemma}\label{lemma:ia}
Suppose $T_1,T_2,\cdots,T_N$ are algebraically independent over the
rational numbers. For any $n \in \mathbb{N}$, we can construct a
$1\times n^N$ vector $\mathbf{V}$ whose entries are rationally
independent and a $1\times (n+1)^N$ vector $\mathbf{U}$ whose
entries are rationally independent as well, such that the following
relations are satisfied.
\begin{eqnarray*}
\textrm{span}(T_1 \mathbf{V}) &\subset& \textrm{span}(\mathbf{U})\\
\textrm{span}(T_2 \mathbf{V}) &\subset& \textrm{span}(\mathbf{U})\\
&\vdots&\\
\textrm{span}(T_N \mathbf{V}) &\subset& \textrm{span}(\mathbf{U})
\end{eqnarray*}
\proof Let us construct two sets $V$ and $U$ with cardinality $n^N$
and $(n+1)^N$, respectively, as follows:
\begin{eqnarray}
V&=&\big\{\prod_{i=1,\ldots, N}T_i^{\alpha_i}: (\alpha_1,\ldots,
\alpha_N) \in \{1,2\ldots, n\}^N \big\}\\
U&=&\big\{\prod_{i=1,\ldots, N}T_i^{\alpha_i}: (\alpha_1,\ldots,
\alpha_N) \in \{1,2\ldots, n+1\}^N \big\}
\end{eqnarray}
Since elements of $V$ are distinct monomials, they are rationally
independent. Similarly, elements of $U$ are rationally independent.
Let entries of $\mathbf{V}$ and $\mathbf{U}$ be the elements of $V$
and $U$, respectively. It can be easily seen that such construction
satisfies all conditions stated above. \hfill\QED
\end{lemma}
Note that the span of a vector here represents the set of all real
numbers that can be represented as linear combinations of the
elements of the vector \emph{with rational coefficients}.

It is important to note that the construction of $V$ and $U$
requires the commutative property of multiplication of numbers
$T_i$. For vector space alignment schemes in interference networks
and wireless $X$ networks \cite{Cadambe_Jafar_int, Cadambe_Jafar_X},
$T_i$ are diagonal matrices which satisfy commutative property of
multiplication as well. Notice that as $n \rightarrow \infty$,
$\frac{|V|}{|U|}\approx 1$. This implies that these two sets are
asymptotically perfectly aligned.

\subsection{Degrees of Freedom of Compound Broadcast Channel}

\begin{figure}[!t]
\centering
\includegraphics[width=3in]{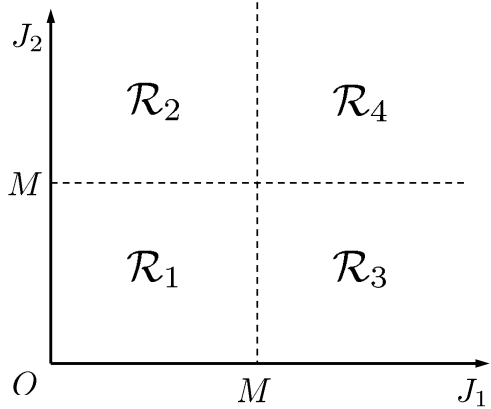}
\caption{$4$ regions for $2$ User Compound Broadcast
Channel}\label{fig:4regions}
\end{figure}

In this section, we first consider the 2 user compound broadcast
channel with $J_1$ and $J_2$ states at each user, respectively.
First, according to the relationship between $J_i$ and $M$, we
partition the $J_1$ and $J_2$ plane into four distinct regions as
illustrated in Figure \ref{fig:4regions}. It can be seen that in the
first region $\mathcal{R}_1$ where $J_1<M$ and $J_2<M$, each user
can achieve 1 degrees of freedom \cite{Weingarten_Shamai_Kramer}.
This can be done by transmitting a data stream along a beamforming
vector orthogonal to channels of the other user. Thus, in this
region, each user achieves the same degrees of freedom as
non-compound broadcast channel. In other words, no degrees of
freedom is lost due to multiple states at each user.

Next we consider regions $\mathcal{R}_2$ and $\mathcal{R}_3$ in
which the number of states for at least one user is no less than $M$
while the other user has less than $M$ states. We establish the
total number of degrees of freedom for these two regions in the
following theorem.
\begin{theorem}\label{thm:2usercompoundbc}
For the compound real broadcast channel with $M$ antennas at the
transmitter, $2$ single antenna users with $J_1<M$ and $J_2 \ge M$
or $J_2<M$ and $J_1 \ge M$, the exact number of total degrees of
freedom is $1+\frac{M-1}{M}$ almost surely.
\end{theorem}
\proof The outer bound follows from \cite{Weingarten_Shamai_Kramer}.
Due to symmetry, let us consider the case when $J_1<M$ and $J_2 \ge
M$. We will show user 1 can achieve 1 DoF while user 2 achieves
$\frac{M-1}{M}$ DoF. First, note that when $J_2=M$, this can be
achieved using zero-forcing at the transmitter
\cite{Weingarten_Shamai_Kramer}. Now consider $J_2>M$. For
simplicity, let us consider the case when $M=2$, $J_1=1$ and
$J_2=3$. The proof for the general case is presented in the
Appendix. Thus, we need to show that user 1 and 2 can achieve 1 and
$\frac{1}{2}$ degrees of freedom, respectively. Note that the linear
alignment solution presented previously for the complex channel does
not apply here, with generic real channel coefficients.

Message $W^{[1]}$ intended for user 1 is split into $2$ sub-messages
denoted as $W^{[1]}_1$ and $W^{[1]}_2$. $W^{[1]}_1$ is encoded into
$n^6$ data streams denoted as $X^{[1]}_{1k},~ k=1,\ldots, n^6$.
$W^{[1]}_2$ is encoded into $n^6$ data streams denoted as
$X^{[1]}_{2k},~ k=1,\ldots, n^6$. Message for user 2 denoted as $W
^{[2]}$ is encoded into $n^6$ independent data streams denoted as
$X^{[2]}_{k},~ k=1,\ldots, n^6$. For any $\epsilon>0$, let
$\mathcal{C}=\{x: x\in\mathbb{Z}\cap [-
P^{\frac{1-\epsilon}{2(m_n+\epsilon)}},
P^{\frac{1-\epsilon}{2(m_n+\epsilon)}}]\}$ where
$m_n=1+(n+1)^{6}+n^{6}$. In other words, $\mathcal{C}$ denotes a set
of all integers in the interval $[-
P^{\frac{1-\epsilon}{2(m_n+\epsilon)}},
P^{\frac{1-\epsilon}{2(m_n+\epsilon)}}]$. Each symbol in the data
stream is obtained by uniformly i.i.d. sampling $\mathcal{C}$.
Essentially, each data stream carries $\frac{1}{m_n}$ degrees of
freedom.

\begin{figure}[!t]
\centering
\includegraphics[width=5in, trim=0 130  0 120]{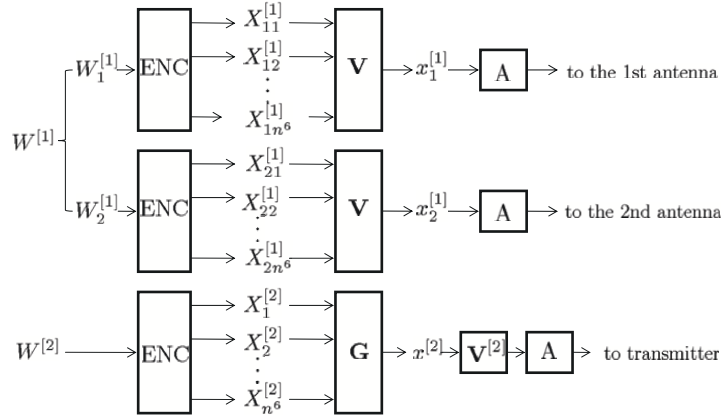}
\caption{Encoding at transmitter 1 and 2 for 2 user compound
broadcast channel}\label{fig:compoundbcenc}
\end{figure}
A data stream $x^{[1]}_i, \forall i=1,2$ is obtained by multiplexing
data streams $X^{[1]}_{ik}, \forall k=1,\ldots, n^{6}$ using a $1
\times n^6$ vector $\mathbf{V}$, i.e,
$x^{[1]}_i=\mathbf{V}\mathbf{X}^{[1]}_i$, where
$\mathbf{V}=[V_1,\ldots, V_{n^6}]$ and
$\mathbf{X}^{[1]}_i=[X^{[1]}_{i1},\ldots, X^{[1]}_{in^6}]^T$. Note
that all elements of $\mathbf{V}$ are functions of channel
coefficients which will be designed to align interference. A data
stream $x^{[2]}$ is obtained by multiplexing $X^{[2]}_k,~
k=1,\cdots, n^{6}$ using a $1 \times n^6$ vector $\mathbf{G}$, i.e,
$x^{[2]}=\mathbf{G}\mathbf{X}^{[2]}$, where
$\mathbf{X}^{[2]}=[X^{[2]}_{1},\ldots, X^{[2]}_{n^6}]^T$. Let
$\mathbf{G}=[G_0~G^2_0~\cdots~G^{n^6}_0]$ where $G_0$ is a randomly
generated real number which is algebraically independent with all
channel coefficients over rationals almost surely. After scaling
with a factor $A$, $x^{[2]}$ is transmitted with a beamforming
vector $\mathbf{V}^{[2]}$ and $x^{[1]}_i$ is transmitted from the
$i$th antenna as illustrated in Figure \ref{fig:compoundbcenc}.
Thus, the transmitted signal is
\begin{eqnarray}
\mathbf{x}=A(\mathbf{V}^{[2]}x^{[2]}+\mathbf{X}^{[1]})
\end{eqnarray}
where  $\mathbf{X}^{[1]}=[x^{[1]}_1~ x^{[1]}_2]^T$ and
$\mathbf{V}^{[2]}$ with unit norm is orthogonal to the channel of
user 1. Thus, no interference is created at user 1. $A$ is a scalar
which is chosen such that the power constraint is satisfied, i.e.,
\begin{eqnarray}
E[\|\mathbf{x}\|^2]&=&E[A(\mathbf{V}^{[2]}x^{[2]}+\mathbf{X}^{[1]})^TA(\mathbf{V}^{[2]}x^{[2]}+\mathbf{X}^{[1]})]\notag\\
&=&
A^2\big(E[(x^{[2]})^2]+E[(x^{[1]}_1)^2]+E[(x^{[1]}_2)^2]\big)\notag\\
&\leq& A^2
\underbrace{(\|\mathbf{G}\|^2+2\|\mathbf{V}\|^2)}_{\lambda^2}P^{\frac{1-\epsilon}{m_n+\epsilon}}\notag\\
&\leq& P\\
\Rightarrow A &=&
\frac{1}{\lambda}P^{\frac{m_n+2\epsilon-1}{2(m_n+\epsilon)}}
\end{eqnarray}

Let us first consider user 2. The received signal at receiver 2
under state $j_2$  is given by
\begin{eqnarray}
y^{[2]}_{j_2} &=& A(\mathbf{h}^{[2]}_{j_2}(\mathbf{V}^{[2]}x^{[2]}+\mathbf{X}^{[1]}))+z^{[2]}_{j_2} \notag\\
          &=&A(
          \underbrace{\mathbf{h}^{[2]}_{j_2}\mathbf{V}^{[2]}}_{h^{[2]'}_{j_2}}x^{[2]}+\mathbf{h}^{[2]}_{j_2}\mathbf{X}^{[1]})+z^{[2]}_{j_2}\notag\\
          &=&A(h^{[2]'}_{j_2}\mathbf{G}\mathbf{X}^{[2]}+h^{[2]}_{j_21}\mathbf{V}\mathbf{X}^{[1]}_1+h^{[2]}_{j_22}\mathbf{V}\mathbf{X}^{[1]}_2)+z^{[2]}_{j_2},~~\forall j_2=1,2,3
\end{eqnarray}
where $\mathbf{h}^{[2]}_{j_2}=[h^{[2]}_{j_21}~h^{[2]}_{j_22}]$. In
order for desired signal $\mathbf{X}^{[2]}$ to get $n^{6}$
interference free dimensions in a total of $1+(n+1)^6+n^6$
dimensional space, we align all interference into a  $(n+1)^6$
dimensional subspace which is spanned by the members of a $1\times
(n+1)^6$ vector $\mathbf{U}$:
\begin{eqnarray}
\textrm{span}(h^{[2]}_{j_2i}\mathbf{V}) &\subset&
\textrm{span}(\mathbf{U}),~~j_2=1,2,3~ ~i=1,2
\end{eqnarray}
From Lemma \ref{lemma:ia}, we construct $\mathbf{V}$ and
$\mathbf{U}$ with rationally independent members to satisfy above
equations. Since $\mathbf{G}$ is generated independently with
$\mathbf{\mathbf{U}}$, members of $h^{[2]'}_{j_2}\mathbf{G}$ and
$\mathbf{U}$ are all distinct and none of them is equal to 1. Thus,
user 2 can achieve $\frac{n^{6}}{1+n^{6}+(n+1)^{6}}$ degrees of
freedom regardless of the realization of the channel almost surely.
As $n\rightarrow \infty$, $\frac{1}{2}$ degrees of freedom can be
achieved.

Now consider user 1. Since there is no interference at user 1, all
data streams are received interference free and along elements
$h^{[1]}_1\mathbf{V}$ and $h^{[1]}_2\mathbf{V}$ where $h^{[1]}_i$ is
the channel coefficient from the $i$th antenna to user 1. It can be
easily seen that members of $h^{[1]}_1\mathbf{V}$ and
$h^{[1]}_2\mathbf{V}$ are all distinct since $\mathbf{V}$ is
independent of  $h^{[1]}_i$. In addition, none of them is 1. Notice
that there are a total of $2n^{6}$ data streams. Since each stream
carries $\frac{1}{1+(n+1)^{6}+n^{6}}$ degrees of freedom, user 1
achieves a total of $\frac{2n^6}{1+(n+1)^{6}+n^{6}}$ DoF almost
surely. As $n\rightarrow \infty$, 1 DoF can be achieved.
 \hfill \QED

As in the complex setting, this is a surprising result. Intuitively,
the DoF will decrease as the number of states associated with the
user increases. However, Theorem \ref{thm:2usercompoundbc} shows
that if one user's states are less than $M$ and regardless of the
number of states associated with the other user, $1+\frac{M-1}{M}$
DoF can be achieved. Thus, in regions $\mathcal{R}_2$ and
$\mathcal{R}_3$, there is only a fraction of $\frac{1}{M}$ DoF lost
due to multiple states at users.

Next we establish the degrees of freedom for $\mathcal{R}_4$ by
solving a general case, i.e., a $N \ge 2$ users compound broadcast
channel where each user has no less than $M$ states. The result is
presented in the following theorem.
\begin{theorem}
For the real compound broadcast channel with $M$ antennas at the
transmitter, $N$ single antenna users with $J_i \ge M, i=1,\ldots,
N$ states at user $i$, the total number of degrees of freedom is
$\frac{MN}{M-N+1}$ almost surely.
\end{theorem}
\proof The achievable scheme is based on $M \times N$ compound $X$
channel discussed later in Theorem \ref{thm:compoundX}. Since the
compound $X$ channel is a restricted form of the MISO BC (the
transmit antennas are separated in the $X$ channel), achievable
degrees of freedom for the $X$ channel are also achievable for the
BC.

For the outer bound, we consider the case where $J_1=1$ and $J_i=M,
i=2,\cdots, N$, since adding more states for each user results in
more constraints and hence cannot increase the rates. The bound is
obtained for a degraded broadcast channel by providing receiver $1$
to $N-1$ with all received signals. Let
$\mathbf{Y}^{[i]}=[y^{[i]}_1\cdots y^{[i]}_{J_i}]$ denote the
received signals for all realizations of user $i$. For an auxiliary
random variable $U$, $U-\mathbf{x}-(y^{[1]},
\mathbf{Y}^{[2]},\cdots, \mathbf{Y}^{[N]})-y^{[N]}_i$ forms a Markov
chain. Thus, we have
\begin{eqnarray}
&& R^{[1]}+R^{[2]}+\cdots+R^{[N-1]}\notag\\
&\leq& I(\mathbf{x};y^{[1]}, \mathbf{Y}^{[2]},\cdots, \mathbf{Y}^{[N]}|U)\notag\\
&=& I(\mathbf{x};\mathbf{Y}^{[N]}|U)+I(\mathbf{x};y^{[1]}, \mathbf{Y}^{[2]},\cdots, \mathbf{Y}^{[N-1]}| U,\mathbf{Y}^{[N]})\notag\\
&=& I(\mathbf{x};\mathbf{Y}^{[N]}|U)+o(\log P)\notag\\
&=&
\sum_{i=1}^{M}I(\mathbf{x};y^{[N]}_i|y^{[N]}_1,\ldots,y^{[N]}_{i-1},
U)+o(\log P)\notag\\
&=&\sum_{i=1}^{M}\big(h(y^{[N]}_i|y^{[N]}_1,\ldots,y^{[N]}_{i-1},
U)-h(y^{[N]}_i|y^{[N]}_1,\ldots,y^{[N]}_{i-1}, U,\mathbf{x})\big)+o(\log P)\notag\\
&\leq& \sum_{i=1}^{M}\big(h(y^{[N]}_i| U)-h(y^{[N]}_i|
U,\mathbf{x})\big)+o(\log P)\notag\\
&=& \sum_{i=1}^{M}I(\mathbf{x};y^{[N]}_i| U)+o(\log P)
\label{eqn:outerbound1}
\end{eqnarray}
On the other hand, $\forall i=1,\ldots, M$,
\begin{eqnarray}
R_N &\leq&  I(U;y^{[N]}_i)
\end{eqnarray}
Adding up all these bounds, we have
\begin{eqnarray}\label{eqn:outerbound2}
MR_N \leq \sum_{i=1}^{M}  I(U;y^{[N]}_i)
\end{eqnarray}
Now adding \eqref{eqn:outerbound1} and \eqref{eqn:outerbound2}, we
have
\begin{eqnarray}
R^{[1]}+\cdots+R^{[N-1]}+MR^{[N]}&\leq& \sum_{i=1}^{M}  (I(U;y^{[N]}_i)+I(\mathbf{x};y^{[N]}_i| U))+o(\log P)\notag\\
&=&  \sum_{i=1}^{M}I(\mathbf{x};y^{[N]}_i)+ o(\log P)\notag \\
&\leq&\frac{M}{2} \log P + o(\log P)
\end{eqnarray}
By symmetry, $\forall i=1,\ldots,N-1$, we have
\begin{eqnarray}
R^{[1]}+\cdots+R^{[i-1]}+MR^{[i]}+R^{[i+1]}+\cdots+R^{[N]} \leq
\frac{M}{2}\log P + o(\log P)
\end{eqnarray}
Adding up all such bounds, we have
\begin{eqnarray}
(M+N-1)(R^{[1]}+\cdots+R^{[N]}) &\leq& \frac{MN}{2}\log P + o(\log P)\notag\\
\Rightarrow R^{[1]}+\cdots+R^{[N]} &\leq& \frac{MN}{2(M+N-1)}\log P
+o(\log P)
\end{eqnarray}
Therefore,
\begin{eqnarray}
d^{[1]}+\cdots+d^{[N]} =\lim_{P\rightarrow
\infty}\frac{R^{[1]}+\cdots+R^{[N]}}{\frac{1}{2}\log P} \leq
\frac{MN}{M+N-1}
\end{eqnarray}
\hfill\QED

Thus we have shown that the DoF of the (real) finite state compound
MISO BC do not collapse to 1 as the channel uncertainty (number of
possible states) increases. What is lost is only the DoF benefits of
joint processing at the transmit antennas, without which the MISO BC
reduces to an $X$ network. Note that for large $M$ this loss also
disappears. In other words, for large $M$, the MISO BC with $N$
users and arbitrary number of states at each user, can achieve
$\frac{MN}{M+N-1}\approx N$ DoF which is the maximum DoF possible
with perfect CSIT.

\section{Compound $X$ Channel}
\begin{figure}[!t]
\centering
\includegraphics[width=3in]{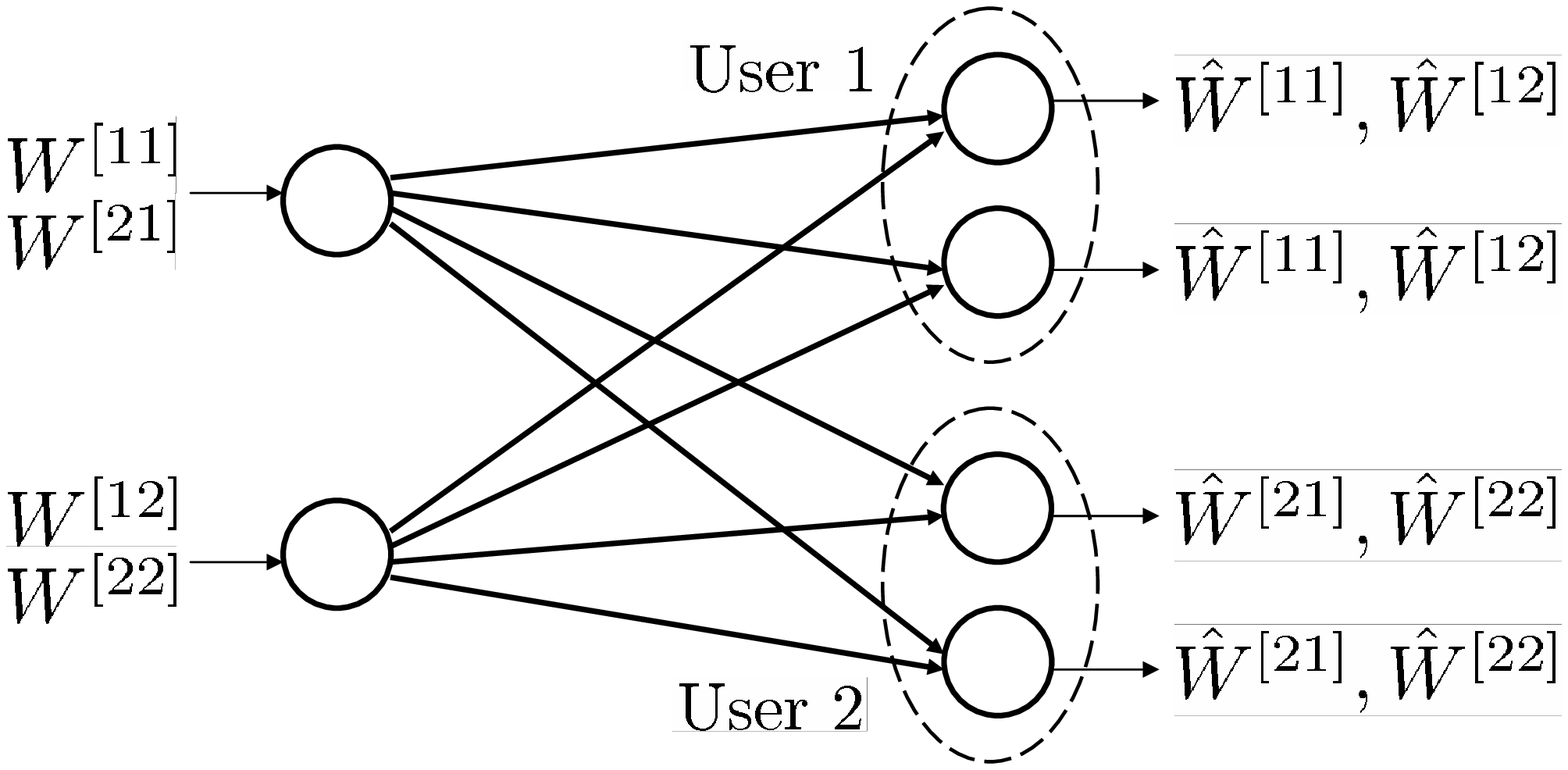}
\caption{$2$ User Compound $X$ Channel with $J_1=J_2=2$}
\label{fig:compoundx2u2s}
\end{figure}
The $M\times N$ wireless compound $X$ channel consists of $M$
transmitters and $N$ receivers. Transmitter $i, \forall i=1,\ldots,
M$ sends an independent message $W^{[ji]}$ with rate $R^{[ji]}$ to
receiver $j, \forall j=1,\ldots, N$. Thus, there are a total of $MN$
messages in the network. Let us denote the channel vector associated
with receiver $j, ~\forall j\in\{1,\ldots, N\}$ as a vector
$(h^{[j1]},\cdots,h^{[jM]})$ which is drawn from a finite set
$\mathcal{J}_j$ with cardinality $J_j$. In addition, we assume the
channel states are drawn from a continuous distribution and hence
algebraically independent over rational numbers almost surely. Once
the channel is drawn, it remains fixed during the entire
transmission. While the transmitters are unaware of the specific
channel state realization, the receivers are assumed to have perfect
channel knowledge. We say the rate tuple $(R^{[11]},\ldots,
R^{[NM]})$ is achievable if all messages can be decoded with
arbitrarily small error probability regardless of the channel
realizations. In this paper, we mainly consider the {\em real}
compound $X$ channel. The received signal at receiver $j$ under
state $k_j$ is given by
\begin{eqnarray}
y^{[j]}_{k_j}=\sum_{i=1}^{M}h^{[ji]}_{k_j}x^{[i]}+z^{[j]}_{k_j}~~
k_j=1,\ldots, J_j,~ j=1,\ldots, N
\end{eqnarray}
where  $h^{[ji]}_{k_j}$ and  $x^{[i]}$  represent the channel
coefficient and transmitted signal, respectively. Transmitter $i$
satisfies the power constraint $E(x_i^2)\leq P$. $z^{[j]}_{k_j}$ is
the additive white Gaussian noise (AWGN) with zero mean and unit
variance. The total number of degrees of freedom, $d$, is defined as
\begin{eqnarray}
d=\lim_{P\rightarrow
\infty}\frac{R^{[11]}+\cdots+R^{[NM]}}{\frac{1}{2}\log P}
\end{eqnarray}
A two user compound $X$ channel with 2 states at each user is shown
in Figure \ref{fig:compoundx2u2s}.

\subsection{Degrees of Freedom of Compound $X$ Network}
We establish the total number of DoF for real compound $X$ network
in the following theorem.

\begin{theorem}\label{thm:compoundX}
For the real compound $M \times N$ $X$ network with $J_j$ states
associated at the $j$th receiver, the total number of degrees of
freedom is $\frac{MN}{M+N-1}$ almost surely.
\end{theorem}
\proof For the non-compound $X$ network, \cite{Cadambe_Jafar_X}
shows that the total degrees of freedom cannot be more than
$\frac{MN}{M+N-1}$. Since compound $X$ network has more decoding
constraints, the outer bound for non-compound $X$ network is also an
outer bound for the compound $X$ network. Next, we provide an
outline of achievable scheme for $2 \times 2$ user $X$ network with
2 states at each user. The detailed proof for general case is
provided in the Appendix.

\begin{figure}[!t]
\centering
\includegraphics[width=4.5in, trim=205 150 0 120]{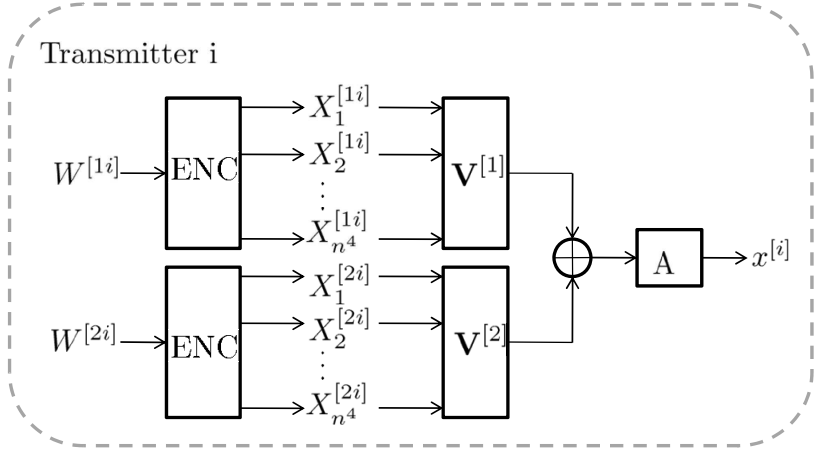}
\caption{Encoding at transmitter $i$} \label{fig:compoundxenc}
\end{figure}
The message from transmitter $i (i=1,2)$, to receiver $j (j=1,2)$
denoted as $W^{[ji]}$, is encoded into $n^4$ independent data
streams. Let $X^{[ji]}_k$ denote the symbol of $k$th data streams
from transmitter $i$ to receiver $j$. For any $\epsilon>0$, let
$\mathcal{C}=\{x: x\in\mathbb{Z}\cap [-
P^{\frac{1-\epsilon}{2(m_n+\epsilon)}},
P^{\frac{1-\epsilon}{2(m_n+\epsilon)}}]\}$ where
$m_n=1+(n+1)^4+2n^4$. Each symbol in the data stream is obtained by
uniformly i.i.d. sampling $\mathcal{C}$. Essentially, each symbol
carries $\frac{1}{m_n}$ degrees of freedom.

At transmitter $i, \forall i=1,2$, the transmitted signal for user
$j$ is obtained by multiplexing different data streams using a
vector $\mathbf{V}^{[j]}$. After scaling with a factor $A$, the
transmitted signal is
\begin{eqnarray}
x^{[i]} &=&
A(\mathbf{V}^{[1]}\mathbf{X}^{[1i]}+\mathbf{V}^{[2]}\mathbf{X}^{[2i]})~~i=1,2
\end{eqnarray}
where $\mathbf{X}^{[ji]}=[X^{[ji]}_1 X^{[ji]}_2 \cdots
X^{[ji]}_{n^4}]^T$ and $\mathbf{V}^{[j]}=[V^{[j]}_1~V^{[j]}_2 \cdots
V^{[j]}_{n^4}] ~\forall j= 1,2$. The encoding at the transmitter is
illustrated in Figure \ref{fig:compoundxenc}. $A$ is a scalar which
is designed such that the power constraints are satisfied, i.e.,
\begin{eqnarray}
E((x^{[i]})^2)\leq P ~~\forall i=1,2
\end{eqnarray}
\begin{eqnarray}
E((x^{[i]})^2)\leq A^2 P^{\frac{1-\epsilon}{m_n+\epsilon}}
\sum_{j=1}^2 \|\mathbf{V}^{[j]}\|^2\leq P
\end{eqnarray}
Let $\lambda^2= \sum_{j=1}^2 \|\mathbf{V}^{[j]}\|^2$ which is a
constant, then
\begin{eqnarray}
A^2P^{\frac{1-\epsilon}{m_n+\epsilon}} \lambda^2 &\leq& P\\
\Rightarrow A &=& \frac{1}{\lambda
}P^{\frac{m_n-1+2\epsilon}{2(m_n+\epsilon)}}
\end{eqnarray}

The received signal for the first state at receiver 1 is given by:
\begin{eqnarray}
y^{[1]}_1 &=&\sum_{i=1}^{2}h^{[1i]}_{1}x_i+z^{[1]}_{1}\notag\\
          &=&A( \underbrace{h_{1}^{[11]}
          \mathbf{V}^{[1]}\mathbf{X}^{[11]}+h^{[12]}_{1}\mathbf{V}^{[1]}\mathbf{X}^{[12]}}_{\textrm{desired
          signal}}+\underbrace{h_{1}^{[11]}\mathbf{V}^{[2]}\mathbf{X}^{[21]}+h_{1}^{[12]}\mathbf{V}^{[2]}\mathbf{X}^{[22]}}_{\textrm{interference}})+z^{[1]}_{1}\notag
\end{eqnarray}
In order to get $2n^4$ interference free dimensions for desired
signal in a total of $1+2n^4+(n+1)^4$ dimensional space, we align
all interference into a $(n+1)^4$ dimensional subspace spanned by
members of $\mathbf{U}^{[2]}$:
\begin{eqnarray}
 \textrm{span}(h^{[11]}_1 \mathbf{V}^{[2]}) &\subset&  \textrm{span}(\mathbf{U}^{[2]})\\
  \textrm{span}(h^{[12]}_1\mathbf{V}^{[2]}) &\subset&
  \textrm{span}(\mathbf{U}^{[2]})
\end{eqnarray}
Similarly, for the second state at receiver 1, we have following
alignment conditions:
\begin{eqnarray}
  \textrm{span}(h^{[11]}_2\mathbf{V}^{[2]}) &\subset&  \textrm{span}(\mathbf{U}^{[2]})\\
  \textrm{span}(h^{[12]}_2\mathbf{V}^{[2]}) &\subset&
  \textrm{span}(\mathbf{U}^{[2]})
\end{eqnarray}
By symmetry, the alignment conditions for user 2 are
\begin{eqnarray}
  \textrm{span}(h^{[22]}_i\mathbf{V}^{[1]}) &\subset&  \textrm{span}(\mathbf{U}^{[1]})\\
  \textrm{span}(h^{[21]}_i\mathbf{V}^{[1]}) &\subset&  \textrm{span}(\mathbf{U}^{[1]})~~i=1,2
 \end{eqnarray}
From Lemma \ref{lemma:ia}, we can construct $\mathbf{V}^{[1]}$,
$\mathbf{V}^{[2]}$, $\mathbf{U}^{[1]}$ and $\mathbf{U}^{[2]}$ to
satisfy those equations. As a result, all interference is received
along members of $\mathbf{U}^{[1]}$ and $\mathbf{U}^{[2]}$ at user 2
and 1, respectively. It can be seen that members of
$\mathbf{V}^{[1]}$ and $\mathbf{V}^{[2]}$ are distinct and
rationally independent. Notice that members of $\mathbf{V}^{[1]}$
and $\mathbf{U}^{[1]}$ depend on $h^{[22]}_i$ and $h^{[21]}_i$,
$\forall i=1,2$ while members of $\mathbf{V}^{[2]}$ and
$\mathbf{U}^{[2]}$ depend on $h^{[12]}_i$ and $h^{[11]}_i$. Thus,
all the desired data streams are received along distinct directions
from the interference and none of them is 1. Thus, each message can
achieve $\frac{n^{4}}{1+(n+1)^{4}+2n^{4}}$ degrees of freedom almost
surely regardless of channel realizations. As $n \rightarrow
\infty$, each message achieves $\frac{1}{3}$ degrees of freedom for
a total of $\frac{4}{3}$ degrees of freedom. \hfill\QED

\noindent{\it Remark:} Theorem \ref{thm:compoundX} also establishes
the total degrees of freedom for the real $M\times N$ wireless $X$
network with constant channel coefficients are $\frac{MN}{M+N-1}$
almost surely. Since this is a special case of compound $X$ network
when each user has only one state. In addition, this indicates that
the finite state compound $X$ channel does not lose any DoF compared
to the non-compound setting.

\section{Compound Interference Channel}
\begin{figure}[!t]
\centering
\includegraphics[width=3in]{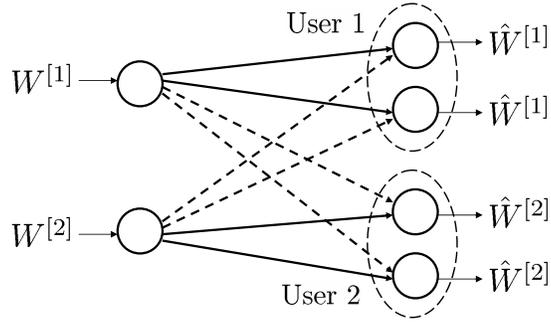}
\caption{$2$ User Compound Interference Channel with $J_1=J_2=2$}
\label{fig:compoundifc2u2s}
\end{figure}
A $K$ user compound interference channel consists of $K$ transmitter
and receiver pairs. Each transmitter sends an independent message
$W^{[j]}, \forall j=1,\ldots, K$ with rate $R^{[j]}$ to its
receiver. Channels associated with receiver $j$ are denoted as the
vector $(h^{[j1]},\cdots,h^{[jM]})$ which comes from a finite set
$\mathcal{J}_j$ with cardinality $J_j$. In addition, we assume the
channel states are drawn from a continuous distribution and hence
algebraically independent over rational numbers almost surely. Once
the channel is drawn, it remains fixed during the entire
transmission. While the transmitters are unaware of the specific
channel state realization, the receivers are assumed to have perfect
channel knowledge. In this section, we consider the real compound
interference channel. The received signal at user $j$ under state
$k_j$ is given by
\begin{eqnarray}
y^{[j]}_{k_j}=\sum_{i=1}^{K}h^{[ji]}_{k_j}x^{[i]}+z^{[j]}_{k_j}~~~
k_j=1,\ldots, J_j,~\forall j=1,\ldots, K
\end{eqnarray}
where $y^{[j]}_{k_j}$, $h^{[ji]}_{k_j}$ and  $x^{[i]}$  represent
the received signal, channel coefficient and transmitted signal,
respectively. Transmitters satisfy the power constraint
$E(x_i^2)\leq P$. $z^{[j]}_{k_j}$ is AWGN with zero mean and unit
variance. We say a rate tuple $(R^{[1]},\ldots,R^{[K]})$ is
achievable if each receiver can decode its message with arbitrarily
small error probability regardless of what state is realized. The
total number of degrees of freedom, $d$, is defined as
\begin{eqnarray}
d=\lim_{P\rightarrow
\infty}\frac{R^{[1]}+\cdots+R^{[K]}}{\frac{1}{2}\log P}
\end{eqnarray}
A two user compound interference channel with 2 states at each user
is shown in Figure \ref{fig:compoundifc2u2s}.

\subsection{Degrees of Freedom of Compound Interference Channel}
Similar to compound $X$ channel, $K$ user interference networks do
not lose DoF in the finite state compound channel setting. We
present the result in the following theorem.
\begin{theorem}\label{thm:compoundic}
The degrees of freedom for $K$ user real compound interference
channel with finite states at each user are $\frac{K}{2}$ almost
surely.
\end{theorem}
\begin{proof}
The proof is provided in the Appendix. \hfill \QED
\end{proof}

\noindent{\it Remark:} Note that if we view different states
associated at each receiver as different users which require
distinct messages from the corresponding transmitter, it is
equivalent to an interference broadcast channel which models the
downlink of cellular network. Specifically, consider a cellular
network with $M$ cells in each of which there are $K$ users. Then
using similar interference alignment schemes used for compound $X$
channel, a total of $\frac{MK}{K+1}$ DoF can be achieved.

\section{Discussion on DoF of Complex Compound Wireless
Networks}\label{sec:discussion}

In Section 3, 4 and 5, we establish the total DoF for different {\em
real} compound wireless networks. A natural question is whether the
same results can be obtained for their complex counterparts. The key
to answer this question is to determine if the interference
alignment schemes used for the real compound interference networks
can be adopted in networks with complex channel coefficients. Since
the alignment schemes are designed for real channel coefficients, we
first write the complex channel as an equivalent channel with real
channel coefficients. Consider a point to point complex channel with
a complex channel coefficient $h$ and a complex noise term $z$. Then
the channel can be written as
\begin{eqnarray}
\left[\begin{array}{c}y_R\\y_I\end{array}\right]=
\left[\begin{array}{cc}h_R & -h_I \\ h_I & h_R \end{array}\right]
\left[\begin{array}{c}x_R\\x_I\end{array}\right]+
\left[\begin{array}{c}z_R \\
z_I\end{array}\right]\label{eqn:complextoreal}
\end{eqnarray}
where subscripts $R$ and $I$ denote the real and imaginary part of
the complex number, respectively. To study the feasibility of the
interference alignment schemes of real wireless networks on their
complex counterparts, let us first consider the non-compound
interference network. In fact, in
\cite{Motahari_Gharan_Maddahali_Khandani}, it is pointed out that
for the $K$ user complex interference channel, every user is still
able to achieve 1/2 DoF. This can be done by viewing  real and
imaginary dimensions  as two independent users. As a result, a user
with complex channel coefficients is converted to a two user real
interference channel where the channel input-output relation is
given by \eqref{eqn:complextoreal}. Thus, instead of a $K$ user
complex interference channel, we obtain a $2K$ user real
interference channel with some dependence among channel coefficients
in the network. Now using the interference alignment scheme on the
real interference network, all interference can be aligned at each
receiver. Note that all beamforming directions are monomials with
variables of all distinct cross links. Since the direct links are
{\em distinct} with all cross links, the directions of desired
signal are distinct of all interference after scaling with the
direct channel, and thus the desired signal is rationally
independent with all interference. Thus, a total of $K$ real DoF,
and hence $K/2$ complex DoF can be achieved. Similarly, the compound
complex interference channel also has $K/2$ DoF almost surely.

However, it is difficult to make the same case for the complex
compound $X$ channel using a similar approach as in the compound
interference channel. It turns out the desired signals overlap with
each other. Intuitively, this is because desired links are the same
as some interference links. Thus, although we can align signals at
receivers where they are not desired, they are aligned at the
desired receivers as well. To avoid the dependence among channel
coefficients, at the cost of nearly half the DoF, we can restrict
the transmitters to send only real signals. For example, for an $M
\times 2$ complex compound $X$ channel where the number of states
associated with each user is greater than $M$, we can obtain an $M
\times 4$ user real compound $X$ channel with all channel
coefficients independent with each other. From Theorem
\ref{thm:compoundX}, $\frac{4M}{M+3}$ real DoF and hence
$\frac{2M}{M+3}$ complex DoF can be achieved. Notice that when
$M>3$, more than 1 DoF can be achieved.  Thus, even for the 2 user
complex MISO broadcast channel,  more than 1 DoF can be achieved
regardless of the number of states as long as the number of antennas
at the transmitter is greater than 3.

\section{Conclusion}\label{sec:5}
This work was motivated by the need to resolve the remarkable
contrast between optimistic results that advocate structured codes
based on the high DoF that can be achieved with perfect channel
knowledge, and pessimistic conjectures that claim that without
perfect channel knowledge the DoF collapse to unity. The strongest
pessimistic conjectures were made by Weingarten et. al. in the
finite state compound channel setting for the MISO BC. In this work
we settle these conjectures in the negative, thereby showing that in
the \emph{finite} state compound channel setting, the DoF results
based on structured codes are robust to channel uncertainty at the
transmitters.

In retrospect, it is perhaps not too surprising that the finite
state compound channel setting does not lose DoF. For example,
consider the $K$ user interference channel. Within this channel,
consider the signals sent by transmitter 1 and 2. In order to
achieve the full $K/2$ DoF, it is clear that these signals must
align at receivers $3,4,\cdots, K$. Clearly, as $K$ increases, i.e.,
more and more receivers are added, bringing new channels into the
picture, the signals from transmitter 1 and 2 must be aligned at
these new receivers while still maintaining alignment at the
previously existing receivers. While it may be surprising at first
to find out that this can be done, it has already been shown in
\cite{Cadambe_Jafar_int}. The compound network setting offers a very
similar challenge. Whatever alignments are needed, must be achieved
for not just one state but for an arbitrary (but finite) number of
states. In the $K$ user interference channel example above, if we
think of the channels to receivers $3,4,\cdots, K$ as multiple
states for the same user, it is clear that the alignment is robust
to the number of states.

The key to the robustness of DoF in the finite state compound
setting is the same as the key to the $K/2$ DoF of the $K$ user
interference channel -- unbounded bandwidth expansion, or
equivalently unlimited resolution in time, frequency, space, or
signal level dimensions. As the alignment problem becomes more and
more challenging, whether by increasing the number of states in the
compound setting or by increasing the number of users in the $K$
user interference channel, greater and greater bandwidth
(equivalently, resolution) is needed to achieve partial alignment.
In the time-varying/frequency-selective $K$ user interference
channel the bandwidth expansion refers to the need to code over
increasingly larger number of symbols. Thinking of these symbols as
frequency slots, we call this a bandwidth expansion. Similar
bandwidth expansion (equivalent to the unbounded resolution of
propagation delays) is observed in the line of sight alignment
schemes found in \cite{Cadambe_Jafar_Delay, Grokop_Tse_Yates}.
Interestingly, when we think of signal level as a signaling
dimension, the unbounded bandwidth expansion or unlimited resolution
essentially corresponds to the infinite precision knowledge of the
channel coefficients. With this infinite precision, we have an
infinite number of signaling level dimensions along which
interference can be aligned regardless of the number of states. Note
that the rational/irrational scaled lattice alignment schemes follow
a complete translation of the Cadambe-Jafar alignment scheme
\cite{Cadambe_Jafar_int} from the time-varying/frequency-selective
channel model to the real constant channel model. Once the role of
rational/irrational scaled lattice alignments is understood in this
context, it is not surprising that finite state compound networks
retain their ability to align signals and thus do not lose their DoF
entirely.

While the DoF are not entirely lost in the finite state compound
setting, it is intriguing that the benefits of transmitter
cooperation are lost. In other words, the MIMO benefits of vector
space alignment are lost. This observation may indicate the distinct
character of alignment schemes over vector spaces and signal levels.
It is notable that inspite of a variety of results on these
different alignment approaches, it has not been possible so far to
unify them into a common framework to understand their collective
synergies and individual limitations.

Another intriguing question that remains open is the implication of
real versus complex models. In all well-studied networks so far,
this issue has been of no real consequence. For interference
alignment problems this issue does become important as it affects
the structure of the channel. However, it is not clear if this is an
avoidable nuisance or if there is a fundamental distinction between
the two settings.

Finally, in the current line of work, the most important issue that
remains unresolved is the robustness of DoF characterizations to
compound networks with infinite states or a continuum of states. In
this regard, the conjecture of Lapidoth et. al.
\cite{Lapidoth_Shamai_Wigger} is most relevant, as is the recent
work on the DoF of the two user MIMO interference channel
\cite{Huang_Jafar_Shamai_Vishwanath}. The overarching observation is
that the best outer bounds known so far are not able to distinguish
between channel uncertainty at the transmitters over a finite set of
states or over a continuum of states. To prove the pessimistic
hypothesis, if indeed the DoF collapse to unity with channel
uncertainty over a continuous (non-zero measure) channel space, then
better outer bounds are needed that can distinguish this setting
from the finite state compound setting. On the other hand, to prove
the optimistic hypothesis, that the DoF are indeed resilient to
channel uncertainty over a continuum, then a much finer
understanding of statistical interference alignment is needed. In
either case, settling this issue will have a profound impact on our
understanding of both the capacity limits of wireless networks as
well as the robustness of these limits.

\appendix
\section{Proof of Theorem {\ref{theorem:firstconj}}}\label{app:firstconj}
\noindent{\it Proof:} The converse is shown in
\cite{Weingarten_Shamai_Kramer}. The achievable scheme is
interference alignment with asymmetric signaling.

Consider the received signal at user $k$ under state $j_k$ in a
single time slot.
\begin{eqnarray}
y_{j_k}^{[k]}={\bf h}_{j_k}^{[k]}{\bf x}+z_{j_k}^{[k]}
\end{eqnarray}
By viewing complex variables as two dimensional vectors, the
received signal can be written as
\begin{eqnarray}
\underbrace{\left[\begin{array}{c}
\mbox{Re}\{y_{j_k}^{[k]}\}\\\mbox{Im}\{y_{j_k}^{[k]}\}\end{array}\right]}_{{\bf
y}_{j_k}^{[k]}:2\times
1}=\underbrace{\left[\!\!\begin{array}{rrrr}\mbox{Re}\{h_{j_k1}^{[k]}\}\!\!&\!\!-\mbox{Im}\{h_{j_k1}^{[k]}\}\!\!&\!\!\mbox{Re}\{h_{j_k2}^{[k]}\}\!\!&\!\!-\mbox{Im}\{h_{j_k2}^{[k]}\}\\\mbox{Im}\{h_{j_k1}^{[k]}\}\!\!&\!\!\mbox{Re}\{h_{j_k1}^{[k]}\}\!\!&\!\!\mbox{Im}\{h_{j_k2}^{[k]}\}\!\!&\!\!\mbox{Re}\{h_{j_k2}^{[k]}\}\end{array}\!\!\right]}_{{\bf
H}_{j_k}^{[k]}:2\times 4}\underbrace{\left[\!\!\begin{array}{c}
\mbox{Re}\{x_1\}\\\mbox{Im}\{x_1\}\\\mbox{Re}\{x_2\}\\\mbox{Im}\{x_2\}\end{array}\!\!\right]}_{{\bf
x}:4\times
1}+\underbrace{\left[\begin{array}{c}\mbox{Re}\{z_{j_k}^{[k]}\}\\\mbox{Im}\{z_{j_k}^{[k]}\}\end{array}\right]}_{{\bf
z}_{j_k}^{[k]}:2\times1}
\end{eqnarray}
Thus we convert the original $1\times 2$ complex MISO BC to a $2
\times 4$ real MIMO BC with a special structure in the channel
matrices. On this new real channel, therefore, we need to show the
achievability of a total of $\frac{3}{2}\times 2=3$ DoF. Due to
$J_1=1$ in this network, we omit the state index of user 1 and
replace $j_2$ with $j$ to denote the state of user 2.

The transmitter sends 2 independent data streams $x_1^{[1]}$ and
$x_2^{[1]}$ to user 1 along with beamforming vectors ${\bf
u}_1^{[1]}$ and ${\bf u}_2^{[1]}$, respectively. In addition, it
sends 1 data stream to user 2 with beamforming vector ${\bf
u}^{[2]}$. Mathematically, we have
\begin{eqnarray}
\overline{{\bf x}}^{[1]}&\!\!\!\!=\!\!\!\!&{\bf u}_1^{[1]}
x_1^{[1]}+{\bf u}_2^{[1]} x_2^{[1]}=\left[{\bf u}_1^{[1]}~{\bf
u}_2^{[1]}\right]\left[\begin{array}{c}x_1^{[1]}\\x_2^{[1]}\end{array}\right]\triangleq{\bf
U}^{[1]}{\bf x}^{[1]}\\
\overline{{\bf x}}^{[2]}&\!\!\!\!=\!\!\!\!&{\bf u}^{[2]} x^{[2]}
\end{eqnarray}
Thus the transmit vector is ${\bf x}=\overline{{\bf
x}}^{[1]}+\overline{{\bf x}}^{[2]}$, and the received signal vectors
of two users are
\begin{eqnarray}
{\bf y}^{[1]}&\!\!\!\!=\!\!\!\!&{\bf H}^{[1]}{\bf
x}+{\bf z}^{[1]}\\
{\bf y}_j^{[2]}&\!\!\!\!=\!\!\!\!&{\bf H}_j^{[2]}{\bf x}+{\bf
z}_j^{[2]}
\end{eqnarray}
At state $j$ of user $2$, we use a $2\times 1$ combining vector
${\bf v}_j^{[2]}$ to get one interference free dimension. Thus the
signal vectors of user 1 and user 2 under state $j$ after linear
combination can be represented as
\begin{eqnarray}
{\bf r}^{[1]}&\!\!\!\!\triangleq \!\!\!\!&{\bf y}^{[1]}={\bf
H}^{[1]}{\bf U}^{[1]}{\bf x}^{[1]}+{\bf H}^{[1]}{\bf
u}^{[2]}x^{[2]}+{\bf z}^{[1]}\\
{\bf r}_j^{[2]}&\!\!\!\!=\!\!\!\!&{\bf v}_j^{[2]T}{\bf
y}_j^{[2]}={\bf v}_j^{[2]T}{\bf H}_j^{[2]}{\bf u}^{[2]}x^{[2]}+{\bf
v}_j^{[2]T}{\bf H}_j^{[2]}{\bf U}^{[1]}{\bf x}^{[1]}+{\bf
v}_j^{[2]T}{\bf z}_j^{[2]}
\end{eqnarray}
In order to decode the desired signals without interference at both
users, we just need to zero force the second item (interference
item) of the two equations above. Thus, our goal is to design ${\bf
U}^{[1]}, {\bf u}^{[2]},{\bf v}_j^{[2]},~j=1,2,3$ such that
following equations are satisfied.
\begin{equation}
\left\{
\begin{array}{rl}
{\bf U}^{[1]T}{\bf H}_j^{[2]T}{\bf v}_j^{[2]}&\!\!\!\!={\bf 0},~~~j=1,2,3\\
{\bf u}^{[2]T}{\bf H}^{[1]T}&\!\!\!\!={\bf 0}
\end{array}
\right.
\end{equation}
To satisfy the first condition, the dimension of the column space of
$\left[\!{\bf H}_1^{[2]T}\!\!\!{\bf v}_1^{[2]}~{\bf
H}_2^{[2]T}\!\!\!{\bf v}_2^{[2]}~{\bf H}_3^{[2]T}\!\!\!{\bf
v}_3^{[2]}\!\right]$ cannot be larger than 2. This is because ${\bf
U}^{[1]T}$ is a $2\times 4$ matrix, which has a 2 dimensional null
space. Since the column spaces of ${\bf H}_1^{[2]T}$ and ${\bf
H}_2^{[2]T}$ only have null intersection, the matrix $\left[{\bf
H}_1^{[2]T}{\bf v}_1^{[2]}~{\bf H}_2^{[2]T}{\bf v}_2^{[2]}\right]$
has rank 2 almost surely. Therefore, we align ${\bf H}_3^{[2]T}{\bf
v}_3^{[2]}$ into the space spanned by column vectors of $\left[{\bf
H}_1^{[2]T}{\bf v}_1^{[2]}~{\bf H}_2^{[2]T}{\bf v}_2^{[2]}\right]$.
To achieve this aim, we first generate ${\bf v}_3^{[2]}$ randomly,
then let
\begin{eqnarray}
\left[\begin{array}{c}{\bf v}_1^{[2]}\\{\bf
v}_2^{[2]}\end{array}\right]={\left[{\bf H}_1^{[2]T}~{\bf
H}_2^{[2]T}\right]}^{-1}{\bf H}_3^{[2]T}{\bf v}_3^{[2]}
\end{eqnarray}
Therefore, we can find 2 linearly independent beamforming vectors of
${\bf U}^{[1]}$ for user 1 that are orthogonal to the column vectors
of $\left[{\bf H}_1^{[2]T}{\bf v}_1^{[2]}~{\bf H}_2^{[2]T}{\bf
v}_2^{[2]}\right]$ and 1 vector ${\bf u}^{[2]}$ for user 2 which is
orthogonal to the column vectors of ${\bf H}^{[1]T}$ such that both
users are free of interference.

What remains to be shown is that at any receiver, the desired signal
vectors after linear combination are linearly independent among
themselves.

First, consider the desired signal after linear combination at user
2 under state $j$ which is given by ${\bf v}_j^{[2]T}{\bf
H}_j^{[2]}{\bf u}^{[2]}x^{[2]}$. We only need to show ${\bf
u}^{[2]T}{\bf H}_j^{[2]T}{\bf v}_j^{[2]}\neq 0$. Note that ${\bf
u}^{[2]}$ can be arbitrarily chosen as any vector orthogonal to the
column vectors of ${\bf H}^{[1]T}$. Thus, ${\bf u}^{[2]T}{\bf
H}_j^{[2]T}{\bf v}_j^{[2]}=0$ implies that ${\bf H}_j^{[2]T}{\bf
v}_j^{[2]}$  lies in the column space of ${\bf H}^{[1]T}$. This,
however, cannot be true since ${\bf H}^{[1]T}, {\bf H}_j^{[2]T}$ are
$4\times 2$ matrices generated i.i.d, the column spaces of ${\bf
H}^{[1]T}, {\bf H}_j^{[2]T}$ only have null intersection almost
surely.

Second, we consider the desired signal of user 1 which is given by
${\bf H}^{[1]}{\bf U}^{[1]}{\bf x}^{[1]}$. To separate two data
streams carried by ${\bf x}^{[1]}$, ${\bf H}^{[1]}{\bf U}^{[1]}$ or
equivalently ${\bf U}^{[1]T}{\bf H}^{[1]T}$ should be a full rank
matrix. Recall that ${\bf U}^{[1]}$ is chosen such that ${\bf
U}^{[1]T}\left[{\bf H}_1^{[2]T}{\bf v}_1^{[2]}~{\bf H}_2^{[2]T}{\bf
v}_2^{[2]}\right]={\bf 0}$. Thus, to show ${\bf H}^{[1]T}$ does lie
in the null space of ${\bf U}^{[1]T}$, we only need to prove the
following matrix has full rank.
\begin{eqnarray}
\left[{\bf H}^{[1]T}~{\bf H}_1^{[2]T}{\bf v}_1^{[2]}~{\bf
H}_2^{[2]T}{\bf v}_2^{[2]}\right]\label{eqn11}
\end{eqnarray}
Since ${\bf h}^{[1]},{\bf h}_1^{[2]},{\bf h}_2^{[2]}$ are three
$1\times 2$ vectors generated i.i.d., we are able to find two
non-zero real coefficients $\beta_1,\beta_2$ such that
\begin{eqnarray}
{\bf h}^{[1]T}=\beta_1{\bf h}_1^{[2]T}+\beta_2{\bf h}_2^{[2]T}
\end{eqnarray}
Considering the mapping from the complex channel to a real channel,
we can see that the matrix ${\bf H}^{[1]}$ can also be linearly
represented by ${\bf H}_1^{[2]},{\bf H}_2^{[2]}$ with the same
coefficients.
\begin{eqnarray}
{\bf H}^{[1]T}=\beta_1{\bf H}_1^{[2]T}+\beta_2{\bf
H}_2^{[2]T}\label{eqn10}
\end{eqnarray}
Since $\left[\!{\bf H}_1^{[2]\!T}\!~\!{\bf H}_2^{[2]\!T}\!\right]$
has full rank almost surely, substituting (\ref{eqn10}) into
(\ref{eqn11}) and multiplying $\left[\!{\bf H}_1^{[2]\!T}\!~\!{\bf
H}_2^{[2]\!T}\!\right]^{\!-\!1}$ to the left hand side of
(\ref{eqn11}) do not change the rank of (\ref{eqn11}). Therefore, we
just need to prove the following matrix has full rank almost surely.
\begin{eqnarray}
\left[\begin{array}{ccc}\beta_1{\bf I}_{2\times 2}&{\bf
v}_1^{[2]}&{\bf O}\\\beta_2{\bf I}_{2\times 2}&{\bf O}&{\bf
v}_2^{[2]}\end{array}\right]
\end{eqnarray}
Recall that
\begin{eqnarray}
\left[\begin{array}{c}{\bf v}_1^{[2]}\\{\bf
v}_2^{[2]}\end{array}\right]={\left[{\bf H}_1^{[2]T}~{\bf
H}_2^{[2]T}\right]}^{-1}{\bf H}_3^{[2]T}{\bf
v}_3^{[2]}=\left[\begin{array}{c}{\bf B}_1{\bf v}_3^{[2]}\\{\bf
B}_2{\bf v}_3^{[2]}\end{array}\right]
\end{eqnarray}
where ${\bf B}_1,{\bf B}_2$ are both $4\times 4$ full rank matrices
with the form
\begin{equation}
\left\{
\begin{array}{rl}
{\bf B}_m&\!\!\!\!=\left[\begin{array}{rr}\mbox{Re}\{b_m\}&-\mbox{Im}\{b_m\}\\ \mbox{Im}\{b_m\}&\mbox{Re}\{b_m\}\end{array}\right]\otimes {\bf I}_{2\times 2},~~~m=1,2\\
b_1&\!\!\!\!=det(\left[{\bf h}_3^{[2]T}~{\bf
h}_2^{[2]T}\right])/det(\left[{\bf h}_1^{[2]T}~{\bf
h}_2^{[2]T}\right]), ~~b_2=det(\left[{\bf h}_1^{[2]T}~{\bf
h}_3^{[2]T}\right])/det(\left[{\bf h}_1^{[2]T}~{\bf
h}_2^{[2]T}\right])
\end{array}
\right.
\end{equation}
where $\otimes$ indicates the Kronecker product operation. The
channels ${\bf h}_j^{[2]}$ are generated i.i.d., thus ${\bf B}_1$ is
not a scaling version of ${\bf B}_2$ almost surely. Scaling the row
vectors and column vectors does not change the rank of a matrix,
therefore, we only need to show that the following matrix has full
rank almost surely.
\begin{eqnarray}
\left[\begin{array}{ccc}{\bf I}_{2\times 2}&{\bf B}_1{\bf
v}_3^{[2]}&{\bf O}\\{\bf I}_{2\times 2}&{\bf O}&{\bf B}_2{\bf
v}_3^{[2]}\end{array}\right]\label{eqn12}
\end{eqnarray}
Let ${\bf \lambda}$ be a $2\times1 $ linear combination vector and
$\lambda_1,\lambda_2$ be two linear combination scalar coefficients.
If the matrix (\ref{eqn12}) has full rank, the following equations
should have only zero solutions.
\begin{eqnarray}
\left[\begin{array}{ccc}{\bf I}_{2\times 2}&{\bf B}_1{\bf
v}_3^{[2]}&{\bf O}\\{\bf I}_{2\times 2}&{\bf O}&{\bf B}_2{\bf
v}_3^{[2]}\end{array}\right]\left[\begin{array}{c}{\bf \lambda}\\
\lambda_1\\ \lambda_2\end{array}\right]={\bf 0}
\end{eqnarray}
Equivalently we can rewrite it as,
\begin{equation}
\left\{
\begin{array}{rl}
{\bf B}_1{\bf v}_3^{[2]}\lambda_1&\!\!\!\!={\bf B}_2{\bf v}_3^{[2]}\lambda_2\\
-{\bf B}_2{\bf v}_3^{[2]}\lambda_2&\!\!\!\!={\bf \lambda}
\label{eqn13}
\end{array}
\right.
\end{equation}
The first equation implies that ${\bf B}_1{\bf v}_3^{[2]}$ and ${\bf
B}_2{\bf v}_3^{[2]}$ are along the same direction. However, this is
not true since ${\bf B}_1$ is not a scalar version of ${\bf B}_2$
almost surely. Thus, the only solution to ({\ref{eqn13}}) is
$\lambda_2=\lambda_3=0$ and ${\bf \lambda}={\bf 0}$. Therefore, all
the column vectors of (\ref{eqn12}) are linearly independent almost
surely. In other words, (\ref{eqn12}) is a full rank matrix almost
surely.

Overall, a total of $\frac{2+1}{2}=\frac{3}{2}$ DoF can be
achievable almost surely. \hfill\QED

\section{Proof of Theorem {\ref{theorem:secondconj}}}\label{app:secondconj}
\noindent{\it Proof:} The converse follows from
\cite{Weingarten_Shamai_Kramer}. The achievable scheme is still
interference alignment with asymmetric signaling.

Consider the $3$ consecutive time slots,
\begin{eqnarray}
\underbrace{\left[\begin{array}{c}y_{j_k}^{[k]}(3n)\\y_{j_k}^{[k]}(3n+1)\\y_{j_k}^{[k]}(3n+2)\end{array}\right]}_{3\times
1}=\underbrace{{\bf h}_{j_k}^{[k]}\otimes {\bf I}_{3\times
3}}_{3\times 6}\underbrace{\left[\begin{array}{c}{\bf x}(3n)\\{\bf
x}(3n+1)\\{\bf x}(3n+2)\end{array}\right]}_{6\times
1}+\underbrace{\left[\begin{array}{c}z_{j_k}^{[k]}(3n)\\z_{j_k}^{[k]}(3n+1)\\z_{j_k}^{[k]}(3n+2)\end{array}\right]}_{3\times
1}
\end{eqnarray}
Thus we have a $3$ dimensional complex signal space, or
equivalently, a $6$ dimensional real signal space.
\begin{eqnarray}
\underbrace{\left[\!\!\begin{array}{c}
\mbox{Re}\{y_{j_k}^{[k]}(3n)\}\\\mbox{Im}\{y_{j_k}^{[k]}(3n)\}\\\mbox{Re}\{y_{j_k}^{[k]}(3n\!\!+\!\!1)\}\\\vdots\\\mbox{Im}\{y_{j_k}^{[k]}(3n\!\!+\!\!2)\}
\end{array}\!\!\right]}_{\overline{{\bf y}}_{j_k}^{[k]}(n):6\times 1}\!\!=\!\!\underbrace{\left[\!\!\begin{array}{rr}\mbox{Re}\{h_{j_k1}^{[k]}\}\!&\!\mbox{Im}\{h_{j_k1}^{[k]}\}\\-\mbox{Im}\{h_{j_k1}^{[k]}\}\!&\!\mbox{Re}\{h_{j_k1}^{[k]}\}\\\mbox{Re}\{h_{j_k2}^{[k]}\}\!&\!\mbox{Im}\{h_{j_k2}^{[k]}\}\\-\mbox{Im}\{h_{j_k2}^{[k]}\}\!&\!\mbox{Re}\{h_{j_k2}^{[k]}\}\end{array}\!\!\right]^T\!\!\!\!\!\!\otimes \!{\bf I}_{3\times 3}\!\!}_{{\bf H}_{j_k}^{[k]}:6\times 12}\underbrace{\left[\!\!\begin{array}{c}
\mbox{Re}\{x_1(3n)\}\\\mbox{Im}\{x_1(3n)\}\\\mbox{Re}\{x_2(3n)\}\\\vdots\\\mbox{Im}\{x_2(3n\!\!+\!\!2))\}\end{array}\!\!\right]}_{\overline{{\bf
x}}(n):12\times
1}\!\!+\!\!\underbrace{\left[\!\!\begin{array}{c}\mbox{Re}\{z_{j_k}^{[k]}(3n)\}\\\mbox{Im}\{z_{j_k}^{[k]}(3n)\}\\\mbox{Re}\{z_{j_k}^{[k]}(3n\!\!+\!\!1)\}\\\vdots\\\mbox{Im}\{z_{j_k}^{[k]}(3n\!\!+\!\!2)\}\end{array}\!\!\right]}_{\overline{{\bf
z}}_{j_k}^{[k]}(n):6\times1}
\end{eqnarray}
After mapping from the complex channel to a real channel, we can
treat it as a MIMO channel with 12 and 6 antennas at the transmitter
and each receiver, respectively. Note that this mapping also
introduces a diagonal structure into the MIMO channel. Therefore, we
need to show the achievability of a total of $\frac{4}{3}\times
2\times 3=8$ DoF for this real channel.

We transmit $4$ data streams to each user. Let ${\bf
u}_m^{[k]},k=1,2~m=1,\ldots,4$ denote the $12\times 1$ beamforming
vector for the $m$-th data stream of user $k$. Then the intended
signal for user $k$ can be represented as
\begin{eqnarray}
\overline{{\bf x}}^{[k]}=\sum_{m=1}^{4}{\bf u}_m^{[k]}
x_m^{[k]}=\left[{\bf u}_1^{[k]}~\ldots~{\bf
u}_4^{[k]}\right]\left[\begin{array}{c}x_1^{[k]}\\\vdots\\x_4^{[k]}\end{array}\right]\triangleq{\bf
U}^{[k]}{\bf x}^{[k]}
\end{eqnarray}
And the transmit signal is $\overline{{\bf x}}=\overline{{\bf
x}}^{[1]}+\overline{{\bf x}}^{[2]}$. Let ${\bf V}_{j_k}^{[k]}$
denote the $6\times 4$ linear combining matrix at user $k$ under
state $j_k$ to achieve 4 interference free dimensions, then the
signal vector after linear combination is
\begin{eqnarray}
{\bf r}_{j_k}^{[k]}={\bf V}_{j_k}^{[k]T}\overline{{\bf
y}}_{j_k}^{[k]}={\bf V}_{j_k}^{[k]T}{\bf H}_{j_k}^{[k]}{\bf
U}^{[1]}{\bf x}^{[1]}+{\bf V}_{j_k}^{[k]T}{\bf H}_{j_k}^{[k]}{\bf
U}^{[2]}{\bf x}^{[2]}+{\bf V}_{j_k}^{[k]T}\overline{{\bf
z}}_{j_k}^{[k]}
\end{eqnarray}
In order for each user to see a clean channel, we need to zero force
the interference items. Equivalently we can write them in the
transpose form.
\begin{equation}
\left\{
\begin{array}{rl}
{\bf U}^{[1]T}{\bf H}_{j_1}^{[2]T}{\bf V}_{j_1}^{[2]}&={\bf 0}~~~~{j_1}=1,2,3\\
{\bf U}^{[2]T}{\bf H}_{j_2}^{[1]T}{\bf V}_{j_2}^{[1]}&={\bf
0}~~~~{j_2}=1,2,3
\end{array}
\right.
\end{equation}
${\bf H}_1^{[k]T}, {\bf H}_2^{[k]T}$ are two $12\times 6$ matrices,
and it can be easily seen that the column spaces of ${\bf
H}_1^{[k]T}$ and ${\bf H}_2^{[k]T}$ only have null intersection
almost surely. Therefore $\left[{\bf H}_1^{[k]T}{\bf V}_1^{[k]}~{\bf
H}_2^{[k]T}{\bf V}_2^{[k]}\right]$ has rank 8 almost surely. Since a
4 dimensional interference free space of the other user should be
protected, we align ${\bf H}_3^{[k]T}{\bf V}_3^{[k]}$ into the
column space of $\left[{\bf H}_1^{[k]T}{\bf V}_1^{[k]}~{\bf
H}_2^{[k]T}{\bf V}_2^{[k]}\right]$. To achieve this goal, we
generate ${\bf V}_3^{[k]}$ randomly, and let
\begin{eqnarray}
\left[\begin{array}{c}{\bf V}_1^{[k]}\\{\bf
V}_2^{[k]}\end{array}\right]={\left[{\bf H}_1^{[k]T}~{\bf
H}_2^{[k]T}\right]}^{-1}{\bf H}_3^{[k]T}{\bf V}_3^{[k]}\label{eqn4}
\end{eqnarray}
Thus we can find 4 linearly independent beamforming vectors to
determine ${\bf U}^{[k]}$ for each user $k$ such that it sees a
clean channel.

What remains to be shown is that at any state of each user, the
desired signal vectors after linear combination are linearly
independent among themselves. Without loss of generality we show
this for user 2. The same argument applies to user 1 due to symmetry
of signaling scheme. Consider the desired signal vector of user 2
under state $j_2$ after linear combination regardless of the noise,
${\bf V}_{j_2}^{[2]T}{\bf H}_{j_2}^{[2]}{\bf U}^{[2]}{\bf x}^{[2]}$.
It is equivalent to a $4\times 4$ MIMO channel, and the matrix ${\bf
U}^{[2]T}{\bf H}_{j_2}^{[2]T}{\bf V}_{j_2}^{[2]}$ should have full
rank almost surely if user 2 can decode its message. Again, since we
have ${\bf U}^{[2]T}\left[{\bf H}_1^{[1]T}{\bf V}_1^{[1]}~{\bf
H}_2^{[1]T}{\bf V}_2^{[1]}\right]={\bf 0}$, our aim can be converted
to prove the following $12\times 12$ matrix has full rank almost
surely.
\begin{eqnarray}
\left[{\bf H}_1^{[1]T}{\bf V}_1^{[1]}~{\bf H}_2^{[1]T}{\bf
V}_2^{[1]}~{\bf H}_{j_2}^{[2]T}{\bf
V}_{j_2}^{[2]}\right]~~~j_2=1,2,3\label{eqn0}
\end{eqnarray}
We show this is true for the state $j_2=3$ and $1$, and the same
argument applies to $j_2=2$.

First consider $j_2=3$. Due to structures of ${\bf H}_{j_k}^{[k]}$,
it can be easily seen that ${\bf H}_3^{[2]}$ linearly depends on
${\bf H}_1^{[1]},{\bf H}_2^{[1]}$. Thus we can find two non-zero
scalar coefficients $\beta_1,\beta_2$ such that
\begin{eqnarray}
{\bf H}_3^{[2]T}=\beta_1{\bf H}_1^{[1]T}+\beta_2{\bf
H}_1^{[2]T}\label{eqn1}
\end{eqnarray}
Again since $\left[{\bf H}_1^{[1]T}~{\bf H}_2^{[1]T}\right]$ has
full rank, substituting (\ref{eqn1}) into (\ref{eqn0}) and
multiplying $\left[{\bf H}_1^{[1]T}~{\bf H}_2^{[1]T}\right]^{-1}$ to
the left hand side of (\ref{eqn0}) do not change the rank of
(\ref{eqn0}). Therefore, we equivalently need to prove the following
matrix has full rank almost surely.
\begin{eqnarray}
\left[\begin{array}{ccc}{\bf V}_1^{[1]}&{\bf O}&\beta_1{\bf
V}_3^{[2]}\\{\bf O}&{\bf V}_2^{[1]}&\beta_2{\bf
V}_3^{[2]}\end{array}\right]
\end{eqnarray}
Remember that
\begin{eqnarray}
\left[\begin{array}{c}{\bf V}_1^{[1]}\\{\bf
V}_2^{[1]}\end{array}\right]={\left[{\bf H}_1^{[1]T}~{\bf
H}_2^{[1]T}\right]}^{-1}{\bf H}_3^{[1]T}{\bf
V}_3^{[1]}=\left[\begin{array}{c}{\bf B}_1{\bf V}_3^{[1]}\\{\bf
B}_2{\bf V}_3^{[1]}\end{array}\right]
\end{eqnarray}
where ${\bf B}_1,{\bf B}_2$ are both $6\times 6$ full rank matrices
with the form
\begin{equation}
\left\{
\begin{array}{rl}
{\bf B}_m&\!\!\!\!=\left[\begin{array}{rr}\mbox{Re}\{b_m\}&-\mbox{Im}\{b_m\}\\ \mbox{Im}\{b_m\}&\mbox{Re}\{b_m\}\end{array}\right]\otimes {\bf I}_{3\times 3},~~~m=1,2\\
b_1&\!\!\!\!=det(\left[{\bf h}_3^{[1]T}~{\bf
h}_2^{[1]T}\right])/det(\left[{\bf h}_1^{[1]T}~{\bf
h}_2^{[1]T}\right]), ~~b_2=det(\left[{\bf h}_1^{[1]T}~{\bf
h}_3^{[1]T}\right])/det(\left[{\bf h}_1^{[1]T}~{\bf
h}_2^{[1]T}\right])
\end{array}
\right.
\end{equation}
Since the channels ${\bf h}_{j_1}^{[1]}$ are generated i.i.d., the
probability of ${\bf B}_1$ being a scaling version of ${\bf B}_2$ is
zero, and scaling the row vectors and column vectors does not change
the rank of a matrix. Therefore, we only need to show that the
following matrix has full rank almost surely.
\begin{eqnarray}
\left[\begin{array}{ccc}{\bf B}_1{\bf V}_3^{[1]}&{\bf O}&{\bf
V}_3^{[2]}\\{\bf O}&{\bf B}_2{\bf V}_3^{[1]}&{\bf
V}_3^{[2]}\end{array}\right]\label{eqn7}
\end{eqnarray}
Let ${\bf \lambda}_1,{\bf \lambda}_2,{\bf \lambda}_3$ be three
$4\times1 $ linear combination vectors. If the matrix (\ref{eqn7})
has full rank, the following equations should only have zero
solutions.
\begin{eqnarray}
\left[\begin{array}{ccc}{\bf B}_1{\bf V}_3^{[1]}&{\bf O}&{\bf
V}_3^{[2]}\\{\bf O}&{\bf B}_2{\bf V}_3^{[1]}&{\bf
V}_3^{[2]}\end{array}\right]\left[\begin{array}{c}{\bf
\lambda}_1\\{\bf \lambda}_2\\{\bf \lambda}_3\end{array}\right]={\bf
0}
\end{eqnarray}
Equivalently we can rewrite it as,
\begin{equation}
\left\{
\begin{array}{rl}
-{\bf B}_1{\bf V}_3^{[1]}{\bf \lambda}_1={\bf V}_3^{[2]}{\bf \lambda}_3\\
-{\bf B}_2{\bf V}_3^{[1]}{\bf \lambda}_2={\bf V}_3^{[2]}{\bf
\lambda}_3 \label{eqn5}
\end{array}
\right.
\end{equation}
This implies that the vector ${\bf V}_3^{[2]}{\bf \lambda}_3$ lies
in the intersection of column spaces of ${\bf V}_3^{[2]}$, ${\bf
B}_1{\bf V}_3^{[1]}$ and ${\bf B}_2{\bf V}_3^{[1]}$. Mathematically,
we have
\begin{eqnarray}
{\bf V}_3^{[2]}{\bf \lambda}_3\in (span({\bf V}_3^{[2]})\cap
span({\bf B}_1{\bf V}_3^{[1]}))\cap (span({\bf V}_3^{[2]})\cap
span({\bf B}_2{\bf V}_3^{[1]}))\\
\Longrightarrow {\bf V}_3^{[2]}{\bf \lambda}_3\in span({\bf
V}_3^{[2]})\cap span({\bf B}_1{\bf V}_3^{[1]})\cap span({\bf
B}_2{\bf V}_3^{[1]})\label{eqn6}
\end{eqnarray}
Since ${\bf V}_3^{[1]}$ is generated randomly, it can be easily seen
that the dimension of matrix $\left[{\bf B}_1{\bf V}_3^{[1]}~{\bf
B}_2{\bf V}_3^{[1]}\right]$ is 6 almost surely. Thus the
intersection of two column spaces of ${\bf B}_1{\bf V}_3^{[1]}$ and
${\bf B}_2{\bf V}_3^{[1]}$ has $4+4-6=2$ dimensions. Recall that
${\bf V}_3^{[2]}$ is also chosen randomly and independently with
${\bf B}_1,{\bf B}_2$ and ${\bf V}_3^{[1]}$, we can conclude that
$span({\bf V}_3^{[2]})$ and $span({\bf B}_1{\bf V}_3^{[1]})\cap
span({\bf B}_2{\bf V}_3^{[1]})$ only have null intersection almost
surely. Hence ${\bf \lambda}_3={\bf 0}$. Substituting it back to
({\ref{eqn5}}), we have
\begin{eqnarray}
{\bf \lambda}_1={\bf \lambda}_2={\bf \lambda}_3={\bf 0}
\end{eqnarray}
Therefore, (\ref{eqn7}) is a full rank matrix almost surely.

Second we consider $j_2=1$. Following the similar analysis, we just
need to show the following matrix has full rank almost surely.
\begin{eqnarray}
\left[\begin{array}{ccc}{\bf B}_1{\bf V}_3^{[1]}&{\bf O}&{\bf
V}_1^{[2]}\\{\bf O}&{\bf B}_2{\bf V}_3^{[1]}&{\bf
V}_1^{[2]}\end{array}\right]\label{eqn9}
\end{eqnarray}
Recall again how is ${\bf V}_1^{[2]}$ generated.
\begin{eqnarray}
\left[\begin{array}{c}{\bf V}_1^{[2]}\\{\bf
V}_2^{[2]}\end{array}\right]={\left[{\bf H}_1^{[2]T}~{\bf
H}_2^{[2]T}\right]}^{-1}{\bf H}_3^{[2]T}{\bf
V}_3^{[2]}=\left[\begin{array}{c}{\bf C}_1{\bf V}_3^{[2]}\\{\bf
C}_2{\bf V}_3^{[1]}\end{array}\right]
\end{eqnarray}
where ${\bf C}_1$ is a $6\times 6$ full rank matrix with the form
\begin{equation}
\left\{
\begin{array}{rl}
{\bf C}_1&\!\!\!\!=\left[\begin{array}{rr}\mbox{Re}\{c_1\}&-\mbox{Im}\{c_1\}\\ \mbox{Im}\{c_1\}&\mbox{Re}\{c_1\}\end{array}\right]\otimes {\bf I}_{3\times 3}\\
c_1&\!\!\!\!=det(\left[{\bf h}_3^{[2]T}~{\bf
h}_2^{[2]T}\right])/det(\left[{\bf h}_1^{[2]T}~{\bf
h}_2^{[2]T}\right])
\end{array}
\right.
\end{equation}
Substitute ${\bf V}_1^{[2]}={\bf C}_1{\bf V}_3^{[2]}$ into
({\ref{eqn9}}), and multiplying $\left[{\bf C}_1\otimes {\bf
I}_{2\times 2}\right]^{-1}$ to the left hand side of ({\ref{eqn9}})
does not change its rank. We thus just need to show
\begin{eqnarray}
\left[\begin{array}{ccc}{\bf C}_1^{-1}{\bf B}_1{\bf V}_3^{[1]}&{\bf
O}&{\bf V}_3^{[2]}\\{\bf O}&{\bf C}_1^{-1}{\bf B}_2{\bf
V}_3^{[1]}&{\bf V}_3^{[2]}\end{array}\right]
\end{eqnarray}
has full rank almost surely. This can be easily seen to be true
since $span({\bf V}_3^{[2]})\cap span({\bf C}_1^{-1}{\bf B}_1{\bf
V}_3^{[1]})\cap span({\bf C}_1^{-1}{\bf B}_2{\bf V}_3^{[1]})$ is
only the null vector almost surely.

Overall, we can achieve a total of
$(\frac{4}{2}+\frac{4}{2})\frac{1}{3}=\frac{4}{3}$ DoF almost
surely. \hfill\QED

\noindent{\it Remark:} Note that the similar alignment scheme does
not work if we apply it with symmetric signaling to the original
complex channel with 3 channel extensions. The reason is that even
though signals can still be aligned at the transmitter, the desired
signal are aligned at the receiver as well. To see this, consider
that ({\ref{eqn6}}) can be also obtained in this case, but here
${\bf V}_3^{[1]}$ and ${\bf V}_3^{[2]}$ are two $3 \times 2$ complex
matrices. ${\bf B}_1$ and ${\bf B}_2$ turn out to be in the form of
${\bf B}_m=b_m{\bf I}_{3\times 3}$, hence scalar versions of the
identity matrix. This implies that the intersection of column spaces
of ${\bf B}_1{\bf V}_3^{[1]}$ and ${\bf B}_2{\bf V}_3^{[1]}$ always
has 2 dimensions. Thus, $span({\bf V}_3^{[2]})\cap span({\bf
B}_1{\bf V}_3^{[1]})\cap span({\bf B}_2{\bf V}_3^{[1]})$ always has
1 dimension. In other words, we can always find non-zero vector
${\bf \lambda}_3$ to satisfy ({\ref{eqn6}}) so that (\ref{eqn7}) is
not a full rank matrix. Therefore, the signal vectors at its
intended receiver are linearly dependent among themselves and each
user fails to decode its message.

\section{Some Examples of the Complex Compound MIMO BC}
In Section \ref{sec:complexmiso}, we investigate some cases of the
complex compound MISO BC. The achievable schemes we use to prove
Theorem \ref{theorem:firstconj} and Theorem \ref{theorem:secondconj}
are both interference alignment with asymmetric signaling. Treating
a complex number as a two dimensional vector with real elements, we
have shown that the complex MISO channel can be treated as a real
MIMO channel but the channel matrix has a special rotation
structure. In addition in the achievable scheme of Theorem
\ref{theorem:secondconj}, we also consider the channel extension
such that the channel has a block diagonal structure. If the channel
has no such special structures, i.e. each entry of the channel
matrix is generated i.i.d., the complex compound MISO BC model would
become compound (\emph{generic}) MIMO BC model. Let us consider two
examples of the complex compound MIMO BC.

{\bf Example 1.} For the complex compound MIMO BC with $K=2$ users,
4 antennas at the transmitter, 2 antennas at each receiver, and
$J_1=1$, $J_2=J=3$ generic channel states for user $1,2$
respectively, the exact number of total DoF = 3, almost surely.

{\bf Example 2.} For the complex compound MIMO BC with $K=2$ users,
6 antennas at the transmitter, 3 antennas at each receiver, and
$J_1=J_2=J=3$ generic channel states for each user, the exact number
of total DoF = 4, almost surely.

In fact, after using asymmetric signaling mapping and multiple
channel extensions, the channel models in Theorem
\ref{theorem:firstconj} and Theorem \ref{theorem:secondconj} are as
same as Example 1 and Example 2, respectively, except for the
special structures of the channel. Using the same alignment scheme,
we achieve the DoF stated in two examples above. However, if $J$
increases from 3 to 4, can we still achieve the same DoF with linear
alignment scheme? The following two examples will answer this
question.

{\bf Example 3.} For the complex compound MIMO BC with $K=2$ users,
4 antennas at the transmitter, 2 antennas at each receiver, and
$J_1=1$, $J_2=J=4$ generic channel states for user $1,2$
respectively, a total of 3 DoF can still be achieved, almost surely.

{\bf Example 4.} For the complex compound MIMO BC with $K=2$ users,
6 antennas at the transmitter, 3 antennas at each receiver, and
$J_1=J_2=J=4$ generic channel states for each user, a total of 4 DoF
can still be achieved, almost surely.

Comparing Example 1 (Example 2) with Example 3 (Example 4), the same
DoF are achieved when $J$ increases from 3 to 4. The difference of
the achievable schemes between the case $J=4$ and $J=3$ starts from
how to choose ${\bf V}_3^{[k]}$. Due to the similar analysis for
Example 3 and 4, we only show the achievability for Example 4.

In the model of Example 4, the transmitter still sends 2 data
streams to each user, respectively. In the case $J=3$, we generate
${\bf V}_3^{[k]}$ randomly. In this case, however, we choose ${\bf
V}_3^{[k]}$ in a different way. Let ${\bf B}_i^{[k]}, i=1,\ldots,4$
denote four $3\times 3$ matrices which are determined by
\begin{eqnarray}
\left[\begin{array}{c}{\bf B}_1^{[k]}\\{\bf
B}_2^{[k]}\end{array}\right]={\left[{\bf H}_1^{[k]T}~{\bf
H}_2^{[k]T}\right]}^{-1}{\bf H}_3^{[k]T}\\
\left[\begin{array}{c}{\bf B}_3^{[k]}\\{\bf
B}_4^{[k]}\end{array}\right]={\left[{\bf H}_1^{[k]T}~{\bf
H}_2^{[k]T}\right]}^{-1}{\bf H}_4^{[k]T}
\end{eqnarray}
Then we let
\begin{eqnarray}
span({\bf B}_1^{[k]}{\bf V}_3^{[k]})=span({\bf B}_3^{[k]}{\bf V}_4^{[k]})\\
span({\bf B}_4^{[k]}{\bf V}_4^{[k]})=span({\bf B}_2^{[k]}{\bf
V}_3^{[k]})
\end{eqnarray}
Thus we obtain
\begin{eqnarray}
span({\bf V}_3^{[k]})=span({\bf B}_1^{[k]-1}{\bf B}_3^{[k]}{\bf
B}_4^{[k]-1}{\bf B}_2^{[k]}{\bf V}_3^{[k]})
\end{eqnarray}
This implies that we can choose two eigenvectors of ${\bf
B}_1^{[k]-1}{\bf B}_3^{[k]}{\bf B}_4^{[k]-1}{\bf B}_2^{[k]}$ as the
column vectors of ${\bf V}_3^{[k]}$. After determining ${\bf
V}_3^{[k]}$, we also determine other combining matrices.
\begin{eqnarray}
{\bf V}_1^{[k]}&\!\!\!\!=\!\!\!\!&{\bf B}_1^{[k]}{\bf V}_3^{[k]}\\
{\bf V}_2^{[k]}&\!\!\!\!=\!\!\!\!&{\bf B}_2^{[k]}{\bf V}_3^{[k]}\\
{\bf V}_4^{[k]}&\!\!\!\!=\!\!\!\!&{\bf B}_3^{[k]-1}{\bf
B}_1^{[k]}{\bf V}_3^{[k]}
\end{eqnarray}
It can be easily seen that all column vectors of ${\bf
H}_3^{[k]T}{\bf V}_3^{[k]},{\bf H}_4^{[k]T}{\bf V}_4^{[k]}$ are
aligned in the column space of $\left[{\bf H}_1^{[k]T}{\bf
V}_1^{[k]}~{\bf H}_2^{[k]T}{\bf V}_2^{[k]}\right]$, thus the
dimension of $\left[{\bf H}_1^{[k]T}{\bf V}_1^{[k]}~{\bf
H}_2^{[k]T}{\bf V}_2^{[k]}~{\bf H}_3^{[k]T}{\bf V}_3^{[k]}~{\bf
H}_4^{[k]T}{\bf V}_4^{[k]}\right]$ is 4 almost surely. Therefore, we
can choose beamforming vectors such that no interference is caused
at each user.

Similar to the proof in the case $J=3$ and due to symmetrical
analysis for user 1 and user 2, we only need to prove the following
matrices have full rank almost surely if desired signal vectors are
linearly independent among themselves at each user.
\begin{eqnarray}
\left[{\bf H}_1^{[1]T}{\bf V}_1^{[1]}~{\bf H}_2^{[1]T}{\bf
V}_2^{[1]}~{\bf H}_{j_2}^{[2]T}{\bf
V}_{j_2}^{[2]}\right]~~~j_2=1,\ldots,4\label{eqn3}
\end{eqnarray}
Notice that ${\bf V}^{[2]}_{j_2}$ is designed independent with ${\bf
H}_1^{[1]}$ and ${\bf H}_2^{[1]}$. In addition, ${\bf V}_1^{[1]}$
and ${\bf V}_2^{[1]}$ are independent with ${\bf H}^{[2]}_{j_2}$.
Since all channel matrices do not have special structure,
({\ref{eqn3}}) has full rank almost surely.

\noindent{\it Remark:} In the case $J=4$ of the compound MIMO BC,
${\bf V}_3^{[k]}$ is determined by the eigenvectors of ${\bf
B}_1^{[k]-1}{\bf B}_3^{[k]}{\bf B}_4^{[k]-1}{\bf B}_2^{[k]}$.
Applying the same scheme to the compound MISO broadcast channel
model in Theorem \ref{theorem:secondconj}, we can see that ${\bf
B}_1^{[k]},{\bf B}_2^{[k]},{\bf B}_3^{[k]},{\bf B}_4^{[k]}$ all
become rotation matrices. Thus ${\bf B}_1^{[k]-1}{\bf B}_3^{[k]}{\bf
B}_4^{[k]-1}{\bf B}_2^{[k]}$ which is also a rotation matrix does
not have real eigenvectors almost surely. The same achievable
scheme, therefore, is not applicable to the complex compound MISO
broadcast channel in Section \ref{sec:complexmiso} due to the
special channel structure.

\section{Proof of Theorem \ref{thm:2usercompoundbc}}
\proof Message $W^{[1]}$ intended for user 1 is split into $M$
sub-messages denoted as $W^{[1]}_i, ~i=1,\ldots, M$.  $W^{[1]}_i
~\forall i=1, 2,\ldots, M$ is encoded into $n^{\Gamma}$ data streams
denoted as $X^{[1]}_{ik},~ \forall k=1,\ldots, n^{\Gamma}$ where
$\Gamma=J_2M$. Message for user 2 denoted as $W ^{[2]}$ is encoded
into $(M-1)n^{\Gamma}$ independent data streams $X^{[2]}_k,~\forall
k=1,\cdots, (M-1)n^{\Gamma}$.  For any $\epsilon>0$, let
$\mathcal{C}=\{x: x\in\mathbb{Z}\cap [-
P^{\frac{1-\epsilon}{2(m_n+\epsilon)}},
P^{\frac{1-\epsilon}{2(m_n+\epsilon)}}]\}$ where
$m_n=1+(n+1)^{\Gamma}+(M-1)n^{\Gamma}$. In other words,
$\mathcal{C}$ denotes a set of all integers in the interval $[-
P^{\frac{1-\epsilon}{2(m_n+\epsilon)}},
P^{\frac{1-\epsilon}{2(m_n+\epsilon)}}]$. Each symbol in the data
stream is obtained by uniformly i.i.d. sampling $\mathcal{C}$.

A data stream  $x^{[1]}_i$ is obtained by multiplexing
$X^{[1]}_{ik},~ \forall i=1,2,\ldots, M,~\forall k=1,\ldots,
n^{\Gamma}$ using the same $1\times n^{\Gamma}$ vector $\mathbf{V}$.
Note that all elements of $\mathbf{V}$ are functions of channel
coefficients which will be designed to align interference. A data
stream $x^{[2]}$ is obtained by multiplexing $X^{[2]}_k,~\forall
k=1,\cdots, (M-1)n^{\Gamma}$ using a vector $\mathbf{G}$. Let
$\mathbf{G}=[G_0, G_0^2, \ldots, G_0^{(M-1)n^{\Gamma}}]$ where $G_0$
is a randomly and independently generated real number. Note that
$G_0$ is algebraically independent with all other channel
coefficients over rationals almost surely. In addition, members of
$\mathbf{G}$ are rationally independent. Mathematically, we have
\begin{eqnarray}
x^{[1]}_i&=&\sum_{k=1}^{n^{\Gamma}}V_{k}X^{[1]}_{ik}=\mathbf{V}\mathbf{X}^{[1]}_i, ~\forall i=1,2,\ldots, M\\
x^{[2]}&=&\sum_{k=1}^{(M-1)n^{\Gamma}}G_0^k
X^{[2]}_{k}=\mathbf{G}\mathbf{X}^{[2]}.
\end{eqnarray}
where $\mathbf{V}=[V_1~\cdots~V_{n^{\Gamma}}]$,
$\mathbf{X}^{[1]}_i=[X^{[1]}_{i1}~\cdots~X^{[1]}_{i n^{\Gamma}}]^T$,
and
$\mathbf{X}^{[2]}=[X^{[2]}_{1}~\cdots~X^{[2]}_{(M-1)n^{\Gamma}}]^T$.
After scaling with a factor $A$, $x^{[2]}$ is transmitted with a
beamforming vector $\mathbf{V}^{[2]}$ and $x^{[1]}_i$ is transmitted
from the $i$th antenna (no cooperation is needed among antennas).
Thus, the transmitted signal is
\begin{eqnarray}
\mathbf{x}=A(\mathbf{V}^{[2]}x^{[2]}+\mathbf{X}^{[1]})
\end{eqnarray}
where  $\mathbf{X}^{[1]}=[x^{[1]}_1~\cdots~ x^{[1]}_M]^T$ and
$\mathbf{V}^{[2]}$ with unit norm is chosen such that no
interference is caused at user 1, i.e.,
\begin{eqnarray}
\mathbf{h}^{[1]}_{j_1} \mathbf{V}^{[2]}=0 ~~\forall j_1=1,\ldots,
J_1
\end{eqnarray}
where $\mathbf{h}^{[1]}_{j_1} $ is the row channel vector of user 1.
$A$ is a scalar which is chosen such that the power constraint is
satisfied, i.e.,
\begin{eqnarray}
E[\|\mathbf{x}\|^2]&=&E[A(\mathbf{V}^{[2]}x^{[2]}+\mathbf{X}^{[1]})^TA(\mathbf{V}^{[2]}x^{[2]}+\mathbf{X}^{[1]})]\notag\\
&=&
A^2\big(E[(x^{[2]})^2]+E[(x^{[1]}_1)^2]+\cdots+E[(x^{[1]}_M)^2]\big)\notag\\
&\leq& A^2
\underbrace{(\|\mathbf{G}\|^2+M\|\mathbf{V}\|^2)}_{\lambda^2}P^{\frac{1-\epsilon}{m_n+\epsilon}}\notag\\
&\leq& P\\
\Rightarrow A &\leq&
\frac{1}{\lambda}P^{\frac{m_n+2\epsilon-1}{2(m_n+\epsilon)}}
\end{eqnarray}

Let us first consider user 2. The received signal at receiver 2
under state $j_2$  is given by
\begin{eqnarray}
y^{[2]}_{j_2} &=& A(\mathbf{h}^{[2]}_{j_2}(\mathbf{V}^{[2]}x^{[2]}+\mathbf{X}^{[1]}))+z^{[2]}_{j_2} \notag\\
          &=&A(
          \underbrace{\mathbf{h}^{[2]}_{j_2}\mathbf{V}^{[2]}}_{h^{[2]'}_{j_2}}x^{[2]}+\mathbf{h}^{[2]}_{j_2}\mathbf{X}^{[1]})+z^{[2]}_{j_2}\notag\\
          &=&A(h^{[2]'}_{j_2}\mathbf{G}\mathbf{X}^{[2]}+h^{[2]}_{j_21}\mathbf{V}\mathbf{X}^{[1]}_1+\cdots+h^{[2]}_{j_2M}\mathbf{V}\mathbf{X}^{[1]}_M)+z^{[2]}_{j_2}~~j_2=1,\ldots,
          J_2
\end{eqnarray}
where $\mathbf{h}^{[2]}_{j_2}=[h^{[2]}_{j_21}\cdots
h^{[2]}_{j_2M}]$. In order to get $(M-1)n^{\Gamma}$ interference
free dimensions for user 2 in a total of
$1+(M-1)n^{\Gamma}+(n+1)^{\Gamma}$ dimensional space, we align all
interference into a $(n+1)^{\Gamma}$ dimensional subspace which is
spanned by members of a vector $\mathbf{U}$:
\begin{eqnarray}
\textrm{span}(h^{[2]}_{j_2k}\mathbf{V}) &\subset&
\textrm{span}(\mathbf{U}) ~~~\forall j_2=1,\cdots,
J_2~~k=1,2,\ldots, M
\end{eqnarray}
From Lemma \ref{lemma:ia}, we construct $\mathbf{V}$ and
$\mathbf{U}$ as follows:
\begin{eqnarray}
\mathbf{V}&=&\big\{\prod(h_{j_2k})^{\alpha_{j_2k}}:
\forall \alpha_{j_2k} \in \{1,2,\ldots,n\}, ~ j_2=1,\ldots, J_2, ~ k=1,2,\ldots, M\big\}\\
\mathbf{U}&=&\big\{\prod(h_{j_2k})^{\alpha_{j_2k}}: \forall
\alpha_{j_2k} \in \{1,2,\ldots,n+1\}, ~ j_2=1,\ldots, J_2, ~
k=1,2,\ldots, M\big\}
\end{eqnarray}
After interference alignment, the effective received signal is
\begin{eqnarray}
y^{[2]}_{j_2}=A(h^{[2]'}_{j_2}\mathbf{G}\mathbf{X}^{[2]}+\mathbf{U}\mathbf{\bar{X}}^{[1]})+z^{[2]}_{j_2}
\end{eqnarray}
where each element of the column vector $\mathbf{\bar{X}}^{[1]}$ is
the sum of all interference along the same direction and is an
integer. Since members of $\mathbf{G}$ are generated independently
with $\mathbf{U}$, all members of $h^{[2]'}_{j_2}\mathbf{G}$ and
$\mathbf{U}$ are distinct, but none of them is equal to 1. Thus,
regardless of the state at user 2, it can achieve
$\frac{(M-1)n^{\Gamma}}{1+(n+1)^{\Gamma}+(M-1)n^{\Gamma}}$ DoF.  As
$n\rightarrow \infty$, $\frac{M-1}{M}$ DoF can be achieved.

Now consider the received signal at user 1 under state $j_1$:
\begin{eqnarray}
y^{[1]}_{j_1} &=& A\mathbf{h}^{[1]}_{j_1}\mathbf{x}+z^{[1]}_{j_1}\notag\\
          &=& A\mathbf{h}^{[1]}_{j_1}\mathbf{X}^{[1]}+z^{[1]}_{j_1}\notag\\
          &=&A(h^{[1]}_{j_11}x^{[1]}_1+\cdots+h^{[1]}_{j_1M}x^{[1]}_M)+z^{[1]}_{j_1}\notag\\
          &=&A(h^{[1]}_{j_11}\mathbf{V}\mathbf{X}^{[1]}_1+\cdots+h^{[1]}_{j_1M}\mathbf{V}\mathbf{X}^{[1]}_M)+z^{[1]}_{j_1}~~j_1=1,\ldots,
          J_1.
\end{eqnarray}
where $\mathbf{h}^{[1]}_{j_1}=[h^{[1]}_{j_11} \cdots
h^{[1]}_{j_1M}]$. It can be easily seen that elements of
$h^{[1]}_{j_11}\mathbf{V},h^{[1]}_{j_12}\mathbf{V},\ldots
h^{[1]}_{j_1M}\mathbf{V}$ are all distinct since members of
$\mathbf{V}$ do not contain $h^{[1]}_{j_11},\ldots, h^{[1]}_{j_1M}$.
In addition, none of them is equal to 1. Thus, regardless of the
channel realization at receiver 1, it can achieve
$\frac{Mn^{\Gamma}}{1+(M-1)n^{\Gamma}+(n+1)^{\Gamma}}$ DoF almost
surely. As $n\rightarrow \infty$, 1 DoF can be achieved. \hfill\QED

\section{Proof for Theorem \ref{thm:compoundX}}
\proof The message from transmitter $i, \forall i=1,\ldots, M$ to
receiver $j, \forall j=1,\ldots, N$ denoted as $W^{[ji]}$ is encoded
into $n^{\Gamma_j}$ independent data streams where $\Gamma_j=M(J_1+
J_2+ \cdots+ J_{j-1}+ J_{j+1}+ \cdots+ J_N)$. Let $X^{[ji]}_k$
denote the symbol of $k$th data stream from transmitter $i$ to
receiver $j$. For any $\epsilon>0$, let $\mathcal{C}=\{x:
x\in\mathbb{Z}\cap [- P^{\frac{1-\epsilon}{2(m_n+\epsilon)}},
P^{\frac{1-\epsilon}{2(m_n+\epsilon)}}]\}$ where $m_n=1+
\max_j\big(Mn^{\Gamma_j}+(N-1)(n+1)^{\Gamma_j}  \big)$. In other
words, $\mathcal{C}$ denotes a set of all integers in the interval
$[- P^{\frac{1-\epsilon}{2(m_n+\epsilon)}},
P^{\frac{1-\epsilon}{2(m_n+\epsilon)}}]$. Each symbol in the data
stream is obtained by uniformly i.i.d. sampling $\mathcal{C}$.

At transmitter $i$,  the transmitted signal to receiver $j$ is
obtained by multiplexing different data streams $X^{[ji]}_{b_j},
\forall b_j=1,\ldots, n^{\Gamma_j}$ using a $1\times n^{\Gamma_j}$
vector $\mathbf{V}^{[j]}$. After scaling with a factor $A$, the
transmitted signal at transmitter $i$ is
\begin{eqnarray}
x^{[i]} &=& A\sum_{j=1}^{N}\sum_{b_j=1}^{n^{\Gamma_j}}V^{[j]}_{b_j}X^{[ji]}_{b_j}\\
&=&
A(\mathbf{V}^{[1]}\mathbf{X}^{[1i]}+\cdots+\mathbf{V}^{[N]}\mathbf{X}^{[Ni]})
~~~ i=1,2,\ldots,M
\end{eqnarray}
where $\mathbf{X}^{[ji]}=[X^{[ji]}_1 X^{[ji]}_2 \cdots
X^{[ji]}_{n^{\Gamma_j}}]^T$  and
$\mathbf{V}^{[j]}=[V^{[j]}_1~V^{[j]}_2 \cdots
V^{[j]}_{n^{\Gamma_j}}] ~\forall j= 1, \ldots, N, ~i=1,\ldots,M$.
$A$ is a scalar which is designed such that the power constraints
are satisfied, i.e.,
\begin{eqnarray}
E((x^{[i]})^2)\leq P ~~\forall i=1,2,\cdots, M
\end{eqnarray}
which can be bounded as
\begin{eqnarray}
E((x^{[i]})^2)\leq A^2 P^{\frac{1-\epsilon}{m_n+\epsilon}}
\sum_{k=1}^N \|\mathbf{V}^{[k]}\|^2\leq P
\end{eqnarray}
Let $\lambda^2= \sum_{k=1}^N \|\mathbf{V}^{[k]}\|^2$ which is a
constant, then
\begin{eqnarray}
A^2P^{\frac{1-\epsilon}{m_n+\epsilon}} \lambda^2 &\leq& P\\
\Rightarrow A &=& \frac{1}{\lambda
}P^{\frac{m_n-1+2\epsilon}{2(m_n+\epsilon)}}
\end{eqnarray}

The received signal at receiver $j$  under state $k_j=1,\ldots, J_j$
is given by:
\begin{eqnarray}
y^{[j]}_{k_j} &=&\sum_{i=1}^{M}h^{[ji]}_{k_j}x_i+z^{[j]}_{k_j}\notag\\
          &=&A( \underbrace{\sum_{i=1}^{M}h^{[ji]}_{k_j}\mathbf{V}^{[j]}\mathbf{X}^{[ji]}}_{\textrm{desired
          signal}}+\underbrace{\sum_{l\neq
          j}h_{k_j}^{[j1]}\mathbf{V}^{[l]}\mathbf{X}^{[l1]}+\sum_{l\neq
          j}h_{k_j}^{[j2]}\mathbf{V}^{[l]}\mathbf{X}^{[l2]}+\cdots+\sum_{l\neq
          j}h_{k_j}^{[jM]}\mathbf{V}^{[l]}\mathbf{X}^{[lM]}}_{\textrm{interference}})+z^{[j]}_{k_j}\notag
\end{eqnarray}
In order to get $M n^{\Gamma_j}$ interference free dimensions in a
total of $1+ \max_j\big(Mn^{\Gamma_j}+(N-1)(n+1)^{\Gamma_j}  \big)$
dimensional space, we align all interference into a
$(N-1)(n+1)^{\Gamma_j}$ dimensional subspace which is spanned by
members of $\mathbf{U}^{[1]},\ldots,\mathbf{U}^{[j-1]},
\mathbf{U}^{[j-1]},\ldots, \mathbf{U}^{[N]} $. Thus, we choose
following alignment equations  at receiver $j$ under state $k_j$ :
$\forall i=1,\ldots, M$
\begin{eqnarray}
\left\{\begin{array}{ccc}
\textrm{span}(h^{[ji]}_{k_j}\mathbf{V}^{[1]}) &\subset& \textrm{span}(\mathbf{U}^{[1]})\\
&\vdots&\\
\textrm{span}(h^{[ji]}_{k_j} \mathbf{V}^{[j-1]}) &\subset& \textrm{span}(\mathbf{U}^{[j-1]})\\
\textrm{span}(h^{[ji]}_{k_j} \mathbf{V}^{[j+1]}) &\subset& \textrm{span}(\mathbf{U}^{[j+1]})\\
&\vdots&\\
\textrm{span}(h^{[ji]}_{k_j}\mathbf{V}^{[N]}) &\subset&
\textrm{span}(\mathbf{U}^{[N]})
\end{array}\right.
\end{eqnarray}
This alignment is illustrated in Figure \ref{fig:compoundxia}. As we
can see,  $h^{[ji]}_{k_j}\mathbf{V}^{[1]},\ldots, h^{[ji]}_{k_j}
\mathbf{V}^{[j-1]}, h^{[ji]}_{k_j} \mathbf{V}^{[j+1]},\ldots,
h^{[ji]}_{k_j}\mathbf{V}^{[N]}$ corresponding to the $i$th row in
Figure \ref{fig:compoundxia} are interference from transmitter $i$
at receiver $j$ under state $k_j$. From another perspective,
corresponding to each column in Figure \ref{fig:compoundxia}, all
interference along with $\mathbf{V}^{[r]}$ is aligned with
$\mathbf{U}^{[r]}$ where $r=1,\ldots, j-1, j+1, \ldots, N$. We can
rewrite all above interference alignment conditions as,
\begin{eqnarray}
 \textrm{span}(h^{[ri]}_{k_r} \mathbf{V}^{[j]}) &\subset&
 \textrm{span}(\mathbf{U}^{[j]})
~~  r, j=1,\ldots, N,~ ~ r\neq j,~~  i=1,\ldots, M, ~~ k_r=1,\ldots,
J_r.
\end{eqnarray}
From Lemma \ref{lemma:ia}, we can construct $\mathbf{V}^{[j]}$ and
$\mathbf{U}^{[j]}$, $\forall j=1,\ldots, N$ as follows:
\begin{eqnarray*}
\mathbf{V}^{[j]}&=&\big\{\prod(h^{[ri]}_{k_r})^{\alpha^{[ri]}_{k_r}}:
\forall \alpha^{[ri]}_{k_r} \in \{1,2,\ldots,n\},  r=1,\ldots, N,~ r\neq j,~ i=1,\ldots, M, ~ k_r=1,\ldots, J_r\big\}\\
\mathbf{U}^{[j]}&=&\big\{\prod(h^{[ri]}_{k_r})^{\alpha^{[ri]}_{k_r}}:
\forall \alpha^{[ri]}_{k_r} \in \{1,2,\ldots,n+1\}, r=1,\ldots, N,~
r\neq j ,~ i=1,\ldots, M, ~ k_r=1,\ldots, J_r\big\}
\end{eqnarray*}
Note that $\mathbf{V}^{[j]}$ and  $\mathbf{U}^{[j]}$ have
$n^{\Gamma_j}$ and $(n+1)^{\Gamma_j}$ elements, respectively, where
$\Gamma_j=M(J_1+ J_2+ \cdots+ J_{j-1}+ J_{j+1}+ \cdots+ J_N)$.

\begin{figure}[!t]
\includegraphics[width=6.5in, trim=0 220 0 100]{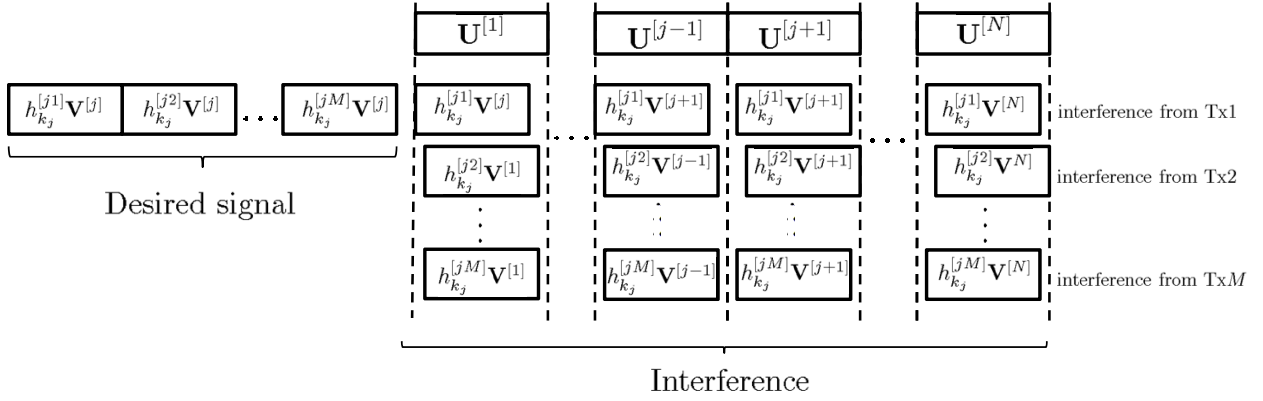}
\caption{Interference alignment at receiver $j$ under state $k_j$
}\label{fig:compoundxia}
\end{figure}
After aligning interference, the equivalent received signal is
\begin{eqnarray}
y^{[j]}_{k_j}=A(\sum_{i=1}^{M}h^{[ji]}_{k_j}
\mathbf{V}^{[j]}\mathbf{X}^{[ji]} +\sum_{l\neq
j}\mathbf{U}^{[l]}\mathbf{X}^{[l]})+z^{[j]}_{k_j}
\end{eqnarray}
where each elements of column vector $\mathbf{X}^{[l]}$ is the sum
of all interference along the same direction.

First note that members of $\mathbf{V}^{[1]},\ldots,
\mathbf{V}^{[N]}$ are distinct. To show that each data stream can
achieve $\frac{1}{m_n}$ DoF, we need to check if all elements of
$h^{[j1]}_{k_j} \mathbf{V}^{[j]}, h^{[j2]}_{k_j}
\mathbf{V}^{[j]},\ldots, h^{[jM]}_{k_j} \mathbf{V}^{[j]}$ and
$\mathbf{U}^{[1]},\ldots,$ $\mathbf{U}^{[j-1]}, \mathbf{U}^{[j+1]},
\ldots, \mathbf{U}^{[N]} $ are distinct. It can be seen that
elements of $h^{[ji]}_{k_j} \mathbf{V}^{[j]}~ \forall i=1,\ldots, M$
are distinct, since $h^{[ji]}_{k_j} $ is not contained in members of
$\mathbf{V}^{[j]}$. In addition, $\mathbf{U}^{[l]}~ \forall l\neq j$
does not have $h^{[li]}_{k_l}, i=1,2,\ldots, M$ while it is
contained in $\mathbf{V}^{[j]}$. Therefore, they are all distinct.
Since none of them is equal to 1, the total number of degrees of
freedom is
\begin{eqnarray}
d&=&\frac{M\sum_{j=1}^N n^{\Gamma_j}}{1+
\max_j\big(Mn^{\Gamma_j}+(N-1)(n+1)^{\Gamma_j}\big)}
\end{eqnarray}
almost surely. As $n \rightarrow \infty$, $d=\frac{MN}{M+N-1}$.
\hfill\QED

\section{Proof of Theorem \ref{thm:compoundic}}
\proof The message from transmitter $i$ to receiver $i$ denoted as
$W^{[i]}$ is encoded into $n^{\Gamma}$ independent data streams
where $\Gamma=(K-1)(J_1+\ldots+J_K)$. Let $X^{[i]}_k$ denote the
symbol of the $k$th data stream from transmitter $i$. For any
$\epsilon>0$, let $\mathcal{C}=\{x: x\in\mathbb{Z}\cap [-
P^{\frac{1-\epsilon}{2(m_n+\epsilon)}},
P^{\frac{1-\epsilon}{2(m_n+\epsilon)}}]\}$ where
$m_n=1+(n+1)^{\Gamma}+n^{\Gamma}$. In other words, $\mathcal{C}$
denotes a set of all integers in the interval $[-
P^{\frac{1-\epsilon}{2(m_n+\epsilon)}},
P^{\frac{1-\epsilon}{2(m_n+\epsilon)}}]$. Each symbol in the data
stream is obtained by uniformly i.i.d. sampling $\mathcal{C}$.

For transmitter $i=1,\ldots, K$, the transmitted signal is obtained
by multiplexing different data streams $X^{[i]}_k, \forall
k=1,\ldots, n^{\Gamma}$ using the same $1\times n^{\Gamma}$ vector
$\mathbf{V}$. Note that all elements of $\mathbf{V}$ are functions
of channel coefficients which will be designed later. Then, the
transmitted signal is
\begin{eqnarray}
x^{[i]} &=& A\sum_{k=1}^{n^{\Gamma}}V_k X^{[i]}_k\\
&=& A\mathbf{V}\mathbf{X}^{[i]} ~~~ i=1,2,\ldots,K
\end{eqnarray}
where $\mathbf{X}^{[i]}=[X^{[i]}_1 X^{[i]}_2 \cdots
X^{[i]}_{n^\Gamma}]^T$ and $\mathbf{V}=[V_1~V_2 \cdots
V_{n^{\Gamma}}]$. $A$ is a scalar which is designed such that the
power constraints are satisfied, i.e.,
\begin{eqnarray}
E(x_i^2)\leq P ~~\forall i=1,2,\cdots, M.
\end{eqnarray}
Since
\begin{eqnarray}
E(x_i^2)\leq A^2 P^{\frac{1-\epsilon}{m_n+\epsilon}}
\|\mathbf{V}\|^2\leq P,
\end{eqnarray}
we have
\begin{eqnarray}
A =
\frac{1}{\|\mathbf{V}\|}P^{\frac{m_n-1+2\epsilon}{2(m_n+\epsilon)}}
\end{eqnarray}

The received signal at receiver $j$ under state $k_j$ is given by:
\begin{eqnarray}
y^{[j]}_{k_j} &=&A(\sum_{i=1}^{K}h^{[ji]}_{k_j}x^{[i]})+z^{[j]}_{k_j}\notag\\
              &=&A( \underbrace{h^{[jj]}_{k_j}\mathbf{V}\mathbf{X}^{[j]}}_{\textrm{desired signal}}+\underbrace{\sum_{i\neq j}h^{[ji]}_{k_j}\mathbf{V}\mathbf{X}^{[i]}}_{\textrm{interference}})+z^{[j]}_{k_j}\
\end{eqnarray}
In order to get $n^{\Gamma}$ interference free dimensions for the
desired signal in a total of $1+n^{\Gamma}+(n+1)^{\Gamma}$
dimension, we align all interference into a $(n+1)^{\Gamma}$
dimensional subspace spanned by members of a  $1\times
(n+1)^{\Gamma}$ vector $\mathbf{U}$:
\begin{eqnarray}
\left\{\begin{array}{ccc}
 \textrm{span}(h^{[j1]}_{k_j}\mathbf{V}) &\subset& \textrm{span}(\mathbf{U})\\
&\vdots& \\
\textrm{span}(h^{[j(j-1)]}_{k_j} \mathbf{V}) &\subset& \textrm{span}(\mathbf{U})\\
\textrm{span}(h^{[j(j+1)]}_{k_j} \mathbf{V}) &\subset& \textrm{span}(\mathbf{U})\\
&\vdots& \\
\textrm{span}(h^{[jK]}_{k_j} \mathbf{V})&\subset&
\textrm{span}(\mathbf{U})
\end{array}\right.
\end{eqnarray}
Equivalently, the above alignment equations can be rewritten as
\begin{eqnarray}
\textrm{span}(h^{[ji]}_{k_j} \mathbf{V}) &\subset& \textrm{span}(
\mathbf{U})~~ i, j=1,\ldots, K,~i\neq j,~k_j=1,\ldots,J_j
\end{eqnarray}
From Lemma \ref{lemma:ia}, we can construct $\mathbf{V}$ and
$\mathbf{U}$ as follows:
\begin{eqnarray}
\mathbf{V}&=&\big\{\prod(h^{[ji]}_{k_j})^{\alpha^{[ji]}_{k_j}}:
\forall \alpha^{[ji]}_{k_j}\in \{1,2,\ldots,n\}, ~ i, j=1,\ldots,
K,~i\neq
j,~k_j=1,\ldots,J_j\big\}\\
\mathbf{U}&=&\big\{\prod(h^{[ji]}_{k_j})^{\alpha^{[ji]}_{k_j}}:
\forall \alpha^{[ji]}_{k_j} \in \{1,2,\ldots,n+1\},~ i, j=1,\ldots,
K,~i\neq j,~k_j=1,\ldots,J_j\big\}
\end{eqnarray}
Note that $\mathbf{V}$ and  $\mathbf{U}$ have $n^{\Gamma}$ and
$(n+1)^{\Gamma}$ elements, respectively, where
$\Gamma=(K-1)(J_1+\ldots+J_K)$. Now after interference alignment,
the received signal is equivalent to
\begin{eqnarray}
y^{[j]}_{k_j}=A(h^{[jj]}_{k_j}\mathbf{V}\mathbf{X}^{[j]}+\mathbf{U}\bar{\mathbf{X}})+z^{[j]}_{k_j}
\end{eqnarray}
where $\bar{\mathbf{X}}$ is a $(n+1)^{\Gamma} \times 1$ vector and
each element of $\bar{\mathbf{X}}$ is the sum of interference along
the same direction. Note that elements of $\mathbf{U}$ do not
contain $h^{[jj]}_{k_j}$ while elements of
$h^{[jj]}_{k_j}\mathbf{V}$ have. Therefore, all elements of
$h^{[jj]}_{k_j}\mathbf{V}$ and $\mathbf{U}$ are distinct. In
addition, none of members of $h^{[jj]}_{k_j}\mathbf{V}$ and
$\mathbf{U}$ is equal to 1. Thus, each user can achieve
$\frac{n^{\Gamma}}{1+n^{\Gamma}+(n+1)^{\Gamma}}$ DoF regardless of
channel realizations almost surely. As $n \rightarrow \infty$, each
user can achieve $\frac{1}{2}$ degrees of freedom almost surely.
 \hfill \QED

\end{document}